\documentclass[conference]{IEEEtran}
\IEEEoverridecommandlockouts
\usepackage[T1]{fontenc}
\usepackage{cite}
\usepackage{amsmath,amssymb,amsfonts}
\usepackage{algorithmic}
\usepackage{graphicx}
\usepackage{textcomp}
\usepackage{xcolor}
\def\BibTeX{{\rm B\kern-.05em{\sc i\kern-.025em b}\kern-.08em
    T\kern-.1667em\lower.7ex\hbox{E}\kern-.125emX}}

\usepackage{custom_packages}

\def\Aoi{\Delta} 

\def\P{\mathbb{P}} 
\def\E{\mathbb{E}} 

\def\hatX{\hat{X}}
\def\hatP{\hat{P}}

\begin{document}

\title{Using Age of Information for Throughput Optimal Spectrum Sharing
}

\author{
Hongjae Nam,~\IEEEmembership{Student Member,~IEEE,}
Vishrant Tripathi,~\IEEEmembership{Member,~IEEE,} 
and David J. Love,~\IEEEmembership{Fellow,~IEEE}
\thanks{H. Nam, V. Tripathi, and D. J. Love are with the Elmore Family School of Electrical and Computer Engineering, Purdue University, West Lafayette, IN 47907 USA (e-mail: \{nam86,  tripathv, djlove\}@purdue.edu). This work was supported in part by the Office of Naval Research (ONR) under Grant N000142112472, and the National Science Foundation (NSF) under Grants CNS2235134, CNS2212565, CNS2225578, and EEC1941529. 
}
}


\maketitle

\begin{abstract}
We consider a spectrum sharing problem where two users attempt to communicate over $N$ channels. The Primary User (PU) has prioritized transmissions and its occupancy on each channel over time can be modeled as a Markov chain. The Secondary User (SU) needs to determine which channels are free at each time-slot and attempt opportunistic transmissions. The goal of the SU is to maximize its own throughput, while simultaneously minimizing collisions with the PU, and satisfying spectrum access constraints. To solve this problem, we first decouple the multiple-channel problem into $N$ single-channel problems. For each decoupled problem, we prove that there exists an optimal threshold policy that depends on the last observed PU occupancy and the freshness of this occupancy information. Second, we establish the indexability of the decoupled problems by analyzing the structure of the optimal threshold policy. Using this structure, we derive a Whittle index-based scheduling policy that allocates SU transmissions using the Age of Information (AoI) of accessed channels. We also extend our insights to PU occupancy models that are correlated across channels and incorporate learning of unknown Markov transition matrices into our policies. 
Finally, we provide detailed numerical simulations that demonstrate the performance gains of our approach.
\end{abstract}


\section{Introduction}
\label{sec:Intro}
\IEEEPARstart{S}{pectrum} sharing plays a fundamental role in modern wireless communication systems, particularly in the context of 6G networks, which aim to achieve ultra-reliable and low-latency communications (URLLC) while efficiently utilizing the radio spectrum. Effective spectrum sharing allows secondary users (SUs) - unlicensed or lower-priority users - to access available spectrum, enabling \textit{opportunistic spectrum} utilization without interfering with primary users (PUs) - licensed or higher-priority users.
This dynamic access approach maximizes spectral efficiency while supporting adaptive and intelligent resource allocation strategies \cite{mitola1999cognitive, ramanathan2005next, brinton2025key}.

\begin{figure}[ht] 
    \centering
    \includegraphics[width=0.9\linewidth]{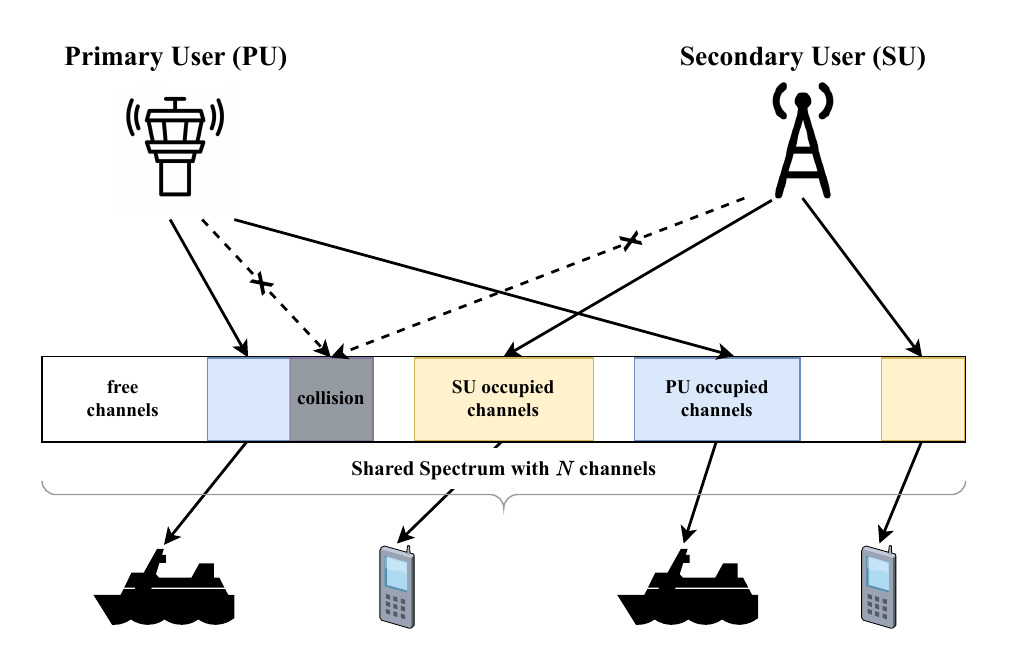}
    \caption{A spectrum sharing scenario with a primary user (PU) and a secondary user (SU) attempting transmissions over $N$ shared channels. The SU needs to carefully decide which channels to utilize so as to avoid collisions with the PU while still maximizing its own throughput.}
    \label{fig:Spectrum sharing scenario}     
\end{figure}

Figure \ref{fig:Spectrum sharing scenario} describes a typical spectrum sharing scenario. The Primary User (PU), such as a naval radar and communications sytem, transmits prioritized information. A Secondary User (SU) tries to utilize unused channels to send lower-priority information.  The Citizens Broadband Radio Service (CBRS), in the 3.55-3.7 GHz frequency range, serves as a representative example of such spectrum sharing. The band consists of three kinds of users: {1) incumbent users such as naval radars or satellite systems, 2) Priority Access License (PAL) users, typically 4G LTE and 5G operators, and 3) General Authorized Access (GAA) users. Incumbent and PAL users (PUs) are protected from interference from GAA users (SUs) through restrictions on simultaneous channel usage. In this tiered access system, the PU has the flexibility to select and utilize its preferred channels, and its spectrum usage pattern may vary dynamically over time. Consequently, the SU must continuously identify open channels for opportunistic spectrum access while minimizing interference with the PU's transmissions.

In this paper, we address how the SU should go about making spectrum access decisions in a dynamic setting while optimizing its own throughput and minimizing collisions with the PU, in the presence of spectrum access constraints. We start by modeling the PU's channel occupancies as independent Markov chains. In this simplified setting, we analyze structural properties of optimal spectrum access policies and design a low-complexity near-optimal scheme that utilizes both historical PU occupancy information and its freshness to decide which channels to utilize next. We then extend our results to more general channel occupancy models that are correlated across channels, and incorporate learning of unknown Markov transition matrices into our policies.


Optimal spectrum access policies, as we will show later, will depend heavily of the Age of Information (AoI) of PU channel occupancy observations. This is because of our Markov modeling assumptions. Recent observations are far more informative than older ones in making decisions about the future. AoI, a metric that quantifies the freshness of information, has gained significant attention in the context of wireless networks \cite{AoI+SS1_2019, AoI+SS2_2019, Tripathi_2014, tripathi2021computation, kadota2016minimizing, kadota2018scheduling}. The interplay between accurate monitoring of channels and throughput maximization leads to an interesting tradeoff: minimizing monitoring error (i.e., AoI) often necessitates transmissions across many channels, while optimizing throughput might mean utilizing the channels that are most likely to be immediately free. Choosing the right balance between exploration and exploitation is key to designing better spectrum access policies.



\paragraph*{\textbf{Related Work}} Spectrum sharing strategies have been extensively studied in the literature, with a focus on maximizing throughput while minimizing interference. Traditional approaches employ multi-armed bandit (MAB) models or a partially observed Markov decision process (POMDP)-based formulations to model the uncertain availability of channels \cite{zhao2007structure, zhao2008myopic, stahlbuhk2016throughput, stahlbuhk2020throughput}. These models have been widely adopted in scenarios where SUs must sequentially decide on sensing and access policies based on limited information. Notably, myopic sensing policies \cite{zhao2007structure, zhao2008myopic} have been shown to achieve near-optimal performance in a two-channel scenario and provided numerical results suggesting their efficiency in multi-channel settings. However, their theoretical optimality has only been proven for independent two-channel cases, and the analysis focuses solely on maximizing throughput without considering collision penalties or access constraints in the presence of hierarchical users. 

Queue-based models have also been employed in analyzing spectrum access and scheduling strategies. In \cite{stahlbuhk2016throughput, stahlbuhk2020throughput}, the authors investigate throughput maximization in uncooperative spectrum-sharing networks, modeling the interaction between primary and secondary users as a queuing system. 
However, their approach 
assumes that the  channel occupancies follow independent Bernoulli processes. This memoryless assumption overlooks temporal correlation in channel states and traffic patterns. Closer to our work, the authors in \cite{liu2010indexability, liu2024efficient} have applied the Whittle index to spectrum sharing. However, they consider discounted reward problems and provide approximations for the Whittle Index. In contrast, our simple Markov model allows us to explore the relationship between AoI and optimal spectrum sharing in much greater detail.


In addition to throughput/queue focused analyses, AoI has also emerged as a key metric for evaluating the freshness of channel state information (CSI) in time-varying channel environments \cite{Farazi_Klein_Brown_2023, farazi2017bounds, klein2017staleness}. The authors highlights that outdated or inaccurate CSI can negatively impact the quality of received signals in fully connected networks. Game-theoretic approaches provide another perspective on spectrum sharing by modeling interactions between SUs as strategic games \cite{zhang2015zero}. The application of zero-determinant (ZD) strategies in spectrum sharing has been explored in \cite{zhang2015zero}, demonstrating their effectiveness in enforcing cooperation and achieving socially optimal outcomes. Unlike traditional MAB-based techniques, ZD strategies enable an SU to exert unilateral control over payoff structures, making them particularly suitable for decentralized environments. Furthermore, recent works on O-RAN-based radar detection in the CBRS band \cite{reus2023senseoran} highlight the potential of leveraging AI-driven sensing techniques to enhance detection accuracy and responsiveness in shared spectrum environments.

\paragraph*{\textbf{Contributions}}
In this work, we propose a novel approach to spectrum sharing that integrates Markov channel occupancies with AoI, allowing us to develop throughput optimal spectrum access policies.
\begin{compactenum}    
    \item We formulate spectrum sharing as an optimization problem that maximizes SU throughput subject to collision penalties and spectrum access constraints.
    \item We decouple this optimization problem into $N$ single-channel problems through a Lagrangian relaxation and demonstrate that the optimal policy has a threshold structure on the AoI of channel access.
    \item We establish indexability of the original problem and use the optimal thresholds to develop the Whittle index policy, enabling a low-complexity method that achieves near-optimal performance. We also show that optimal spectrum sharing has close ties to AoI, since maintaining fresh information regarding channel states allows for better spectrum access.
    \item We further extend our approach to handle correlated PU channel occupancies by designing a heuristic index function inspired by the Whittle index. We also incorporate learning of unknown Markov transition matrices to our policies. We demonstrate the benefits of both the Whittle index as well as the heuristic index through detailed simulations. 
\end{compactenum}

\paragraph*{\textbf{Organization}} The remainder of this paper is organized as follows. In Section~\ref{sec:System model}, we describe the general system model with a probabilistic Markov chain channel framework and introduce the spectrum sharing problem. 
In Section~\ref{sec:Solving Decoupled Prob.}, we solve the resulting decoupled single-channel problem via Bellman equations and propose the optimal threshold type policy. In Section~\ref{sec:Scheduling Policies}, we derive a closed-form expression for the Whittle index using the optimal threshold policy and propose a Whittle index-based scheduling algorithm for the original multi-channel setting. Section~\ref{sec:Correlated Channel} extends our Whittle index approach to scenarios involving inter-channel correlation. In Section~\ref{sec:Numerical results}, we provide simulation results to validate our analysis, followed by the conclusions in Section \ref{sec:Conclusion}. 

\section{System Model} \label{sec:System model}
We consider a time-slotted network with two classes of users trying to transmit information over shared spectrum. The \textit{Primary User} (PU) gets priority when transmitting on any channel and does not have any spectrum access constraints. The \textit{Secondary User} (SU) needs to find channels on which the PU is not currently transmitting, so as to be able to communicate. The SU pays a cost for colliding with the PU or for occupying too many channels. The goal of the SU is to maximize its total time-average throughput, while ensuring that it does not collide with the PU, and satisfy its own energy/spectrum access constraints.

We assume that there are $N$ channels available in total. We model the PU occupancy on each channel through a symmetric two-state Markov chain as shown in Fig. \ref{fig:Markov_Channel}. If channel $i$ is currently occupied by the PU, then in the next time-slot, it transitions to free with probability $q_i$ and remains occupied with probability $1-q_i$. This allows us to model historical dependence of channel occupancy. Specifically, we assume that $q_i < 1/2, \forall i \in [N]$. Thus, if the PU is currently occupying a channel, it is more than likely to still be occupying the channel in the next time-slot. For analysis, we make the following key assumptions in our model:
\begin{assumption} \label{assm: Markov transition prob.}
The SU knows the Markov transition probabilities $\{q_i~|~i\in[N]\}$ for each channel $i$ in the spectrum and the corresponding transition matrices are given by
\begin{align} \label{eq:Transition matrix Q}
Q_i=\left[\begin{array}{cc}
1-q_i & q_i \\
q_i & 1-q_i 
\end{array}\right]
~,\forall i \in [N].
\end{align}
\end{assumption}

\begin{assumption} \label{assm: independence}
The PU channel occupancies, modeled by 2-state Markov chains, evolve \textbf{independently} across different channels.
\end{assumption}

These two assumptions will allow us to perform analysis, provide performance guarantees, and gain important insights into the process of spectrum sharing. Importantly, we will design spectrum access policies for the SU that utilize the knowledge of the transition matrices as well as PU channel occupancy information obtained in the past, to optimize the overall throughput of the SU. In Section \ref{sec:Correlated Channel}, we extend our Whittle index-based approach to scenarios where Assumption \ref{assm: independence} is relaxed, specifically when PU channel occupancies exhibit correlation across channels. In Section \ref{sec:Numerical results}, we also provide numerical results when Assumption \ref{assm: Markov transition prob.} is relaxed, i.e., the transition matrices $Q_i$ are unknown and need to be learned.

\begin{figure}[h] 
    \centering
    \includegraphics[width=0.6\linewidth]{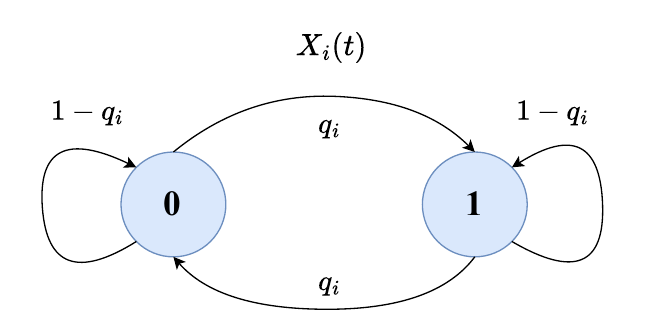}
    \caption{A two-state symmetric Markov chain that models the PU occupancy of the $i$-th channel.}
    \label{fig:Markov_Channel} 
\end{figure}

Let $\mathbf{X}(t) = [X_1(t),X_2(t),...,X_{N}(t)]^T \in \{0,1\}^N$ be a channel state vector representing PU occupancy at the time-slot $t$. For example, $X_i(t)=1$ indicates that the $i$-th channel is occupied by the PU at time-slot $t$, while $X_i(t)=0$ implies that it is vacant and available for the SU to use. To the SU, the $N$ channels appears as if they are independent Gilbert-Elliot channels \cite{gilbert1960capacity}, which are represented using two-state Markov chains to describe their state transitions over time. Under our Markov channel occupancy model, the fraction of time for which the PU occupies each channel converges to 50$\%$ due to the symmetric nature of the Markov transition matrix. The transition probability $q_i$ controls whether each period of PU channel occupancy is relatively short or long, influencing the rate at which PU occupancy switches.

\paragraph*{\textbf{Sensing and Collision Model}} In each time-slot, the SU can decide whether to transmit on channel $i$ or leave it empty. Let $u_i(t) \in \{0,1\}$ indicate whether the SU attempts a transmission on channel $i$ at time-slot $t$. Suppose the SU transmits on channel $i$ at time-slot $t$ ($u_i(t) = 1$). Then, if the channel is already occupied by the PU ($X_i(t)=1$), the SU transmission leads to a collision with the PU and fails. However, if the channel is unoccupied by the PU ($X_i(t) = 0$), the SU successfully transmit one packet. The failure or success of the transmission is revealed to the SU at the beginning of the next time-slot $t+1$ via an ACK. Thus, to sense PU channel occupancies, the SU must attempt regular transmissions on different channels. If the SU does not transmit on a channel for some time, its estimate of channel occupancy starts growing stale and less useful for making spectrum access decisions.



\paragraph*{\textbf{Measuring Throughput}} The SU successfully delivers one packet only when the PU is absent and no collision occurs. Mathematically, let $ R_i(t)$ indicate whether a packet was successfully transmitted by the SU over channel $i$ at time-slot $t$. Then, it is given by
\begin{equation} \label{eq:SU throughput}
    \begin{aligned}
        R_i(t) = u_i(t) (1-X_i(t)),~\forall i \in [N]. 
    \end{aligned}
\end{equation}
Clearly, the SU transmits a packet, i.e., $R_i(t)=1$, when it selects channel $i$ ($u_i(t)=1$) and the PU is absent on that channel ($X_i(t)=0$). Otherwise, the throughput is zero. 

\paragraph*{\textbf{Spectrum Access Constraints}} Due to practical hardware limitations and energy constraints, the SU will typically not be capable of simultaneously sensing the entire set of available channels in the spectrum. Further, even if it were capable of doing so, attempting transmission on all channels will cause significant disruption to the PU. So, we assume that the SU can access at most $L$ channels in each time-slot. Mathematically, we can write this down as a set of constraints on the spectrum access variables $u_i(t)$ as follows:
\begin{equation} \label{eq:Spectrum access constraints}
    \begin{aligned}
        \sum_{i=1}^{N} u_i(t) \leq L, \forall t \in [T],  
    \end{aligned} 
\end{equation}
where $L$ is the maximum number of channels the SU can access at a given time. These spectrum access constraints ensures that the SU needs to make selective and efficient use of its available transmission opportunities. As a result, an efficient sensing strategy is required to intelligently select channels based on past PU occupancy observations.

\paragraph*{\textbf{Collision Penalties}} In our spectrum sharing scenario, the SU may unintentionally select channels that are already occupied by a PU, leading to interference with higher-priority PU transmissions. We explicitly incorporate the impact of collisions as penalty terms in the objective function. 
Mathematically, a collision occurs when the SU selects channel $i$ ($u_i(t)=1$) while the PU is also present on the same channel ($X_i(t)=1$). We use $C_i(t)$ to indicate the collision penalty incurred by the SU on channel $i$ at time-slot $t$.
\begin{equation} \label{eq:Collision penalties}
C_i(t) = \gamma u_i(t) X_i(t), \; \forall i \in [N],
\end{equation}
where $\gamma > 0$ represents the penalty cost for each individual collision. Imposing this penalty ensures that the SU's channel selection strategy considers the impact on both its own throughput and the PU's performance, thereby discouraging channel access that results in harmful interference.

Given the setup described above, we aim to formulate a unified optimization framework that simultaneously accounts for both spectrum access constraints and collision penalties. This framework enables the SU to make optimal channel selections by jointly considering limited access opportunities and interference costs. Our goal is to design a causal spectrum access policy $\pi$ to maximize the total SU throughput under limited spectrum access constraints and collision penalties.

\begin{mdframed}[linewidth=1pt, linecolor=black]
\textbf{Problem 1} \textit{(SU Throughput Maximization under Spectrum Access Constraints and Collision Penalties)}

Given Markov transition probabilities $\{q_i~|~ i \in [N]\}$, a spectrum usage constraint of at most $L$ channels at each time-slot, and a collision penalty of $\gamma>0$ per collision, we want to maximize the long term expected SU throughput minus the collision penalties.
\begin{equation} \label{eq:P1, Unified}
\begin{aligned}
    \max_{\pi} \quad & \lim \limits_{T \to \infty} \mathbb{E}_{\pi} \left[  \frac{1}{L} \frac{1}{T} \sum_{i=1}^{N} \sum_{t=1}^{T} \Big(R_{i}(t)- C_i(t)\Big) \right] \\
   \textrm{s.t.}     
   \quad & \sum_{i=1}^{N} u_i(t) \leq L, \forall t \in [T],\\
   \quad & u_i(t) \in \{0,1\}, \forall i \in [N], \forall t \in [T]. \\   
\end{aligned}    
\tag{P1} 
\end{equation}
\end{mdframed}

Solving \eqref{eq:P1, Unified} directly is challenging due to the spectrum access constraint in each time-slot. As our system model satisfies the Markov property, \eqref{eq:P1, Unified} can be formulated as a restless multi-armed bandit (RMAB) problem with state transitions governed by the PU and the SU. Since the RMAB problem has been proven to be \textit{PSPACE-hard} \cite{papadimitriou1999complexity}, finding an optimal solution to \eqref{eq:P1, Unified} is computationally intractable. This is because our action space in each time-slot is combinatorial over the number of channels and the number of spectrum access policies scales exponentially in the length of the time-horizon, making dynamic programming (DP) based solutions infeasible. To manage this complexity, we first relax the per-time-slot constraint into a time-average constraint. We then use a Lagrangian relaxation \cite{kriouile2024asymptotically} to decouple \eqref{eq:P1, Unified} into $N$ single-channel problems, each of which is much easier to solve optimally using dynamic programming. Building on this optimal policy for a single channel, we propose a low-complexity near-optimal scheduling policy using the Whittle index approach to tackle the original multi-channel problem \eqref{eq:P1, Unified}.

To decouple the problem across channels, we introduce a Lagrangian multiplier $C>0$, which acts as the price for accessing a channel. Normalizing this price defines an effective transmission cost $D=(\gamma+C)/(1+\gamma)$, which leads to the following result.

\begin{mdframed}[linewidth=1pt, linecolor=black]
\begin{thm} \label{thm: Decoupled Singe-Channel Problem}
The spectrum sharing problem \eqref{eq:P1, Unified} decouples in to $N$ single-channel problems of the form below, where $D>0$ is an effective transmission cost.
\begin{equation} \label{eq:P1-d}
\begin{aligned}
    \max_{\pi} \quad & \lim \limits_{T \to \infty} \mathbb{E}_{\pi} \left[  \frac{1}{T}   \sum_{t=1}^{T} \Big(R_{i}(t) - D u_i(t)\Big) \right] \\   
   \textrm{s.t.}     
   \quad & u_i(t) \in \{0,1\}, \forall t \in [T]. \\  
\end{aligned}    
\tag{P1-d}
\end{equation}
\end{thm}
\end{mdframed}

\begin{proof}
By relaxing the per-time-slot access constraints in \eqref{eq:P1, Unified} to a time-average constraint, we apply Lagrangian relaxation with the multiplier $C$. This transformation converts the global constraint into a transmission cost, thereby enabling decoupling across channels. For more details, see Appendix A. 
\end{proof}

Importantly, note that the effective transmission cost $D=(\gamma+C)/(1+\gamma)$ is an increasing function of both the collision penalty $\gamma$ and the transmission cost $C$, which aligns with intuition. When $\gamma=0$, the effective cost reduces to $D=C$, indicating that only the Lagrange transmission cost remains. In contrast, as $\gamma \to \infty$, the effective cost $D$ approaches $1$, since we assume that the SU throughput \eqref{eq:SU throughput} on a single channel is measured as a one-packet success, an infinite collision penalty effectively nullifies the reward by imposing a penalty of 1. By obtaining the optimal policy for each decoupled single-channel problem with an effective transmission cost \eqref{eq:P1-d}, we can optimally solve the overall Lagrangian relaxation. As we will discuss later, this  will also allow us design a spectrum access policy for the original problem \eqref{eq:P1, Unified} that is nearly optimal via the Whittle index framework. 

\section{Solving the Decoupled Problem} \label{sec:Solving Decoupled Prob.}

In this section, we will solve the decoupled single-channel problem \eqref{eq:P1-d} using dynamic programming. The first observation we make is that the decoupled problem can be modeled as a Markov decision process (MDP) with an infinite horizon time-average reward. To describe this single-channel MDP, we need to describe four things - the state space, the action space, the rewards, and the Markov transitions. We describe each of these below.

\paragraph*{\textbf{States}} Let us suppose that at time-slot $t$, the SU last accessed the channel at time-slot $t-\Aoi_i(t)$. Then, it has two pieces of information - 1) the last observed PU occupancy, i.e., $X_i(t-\Aoi_i(t))$ and 2) the staleness of the PU channel occupancy information at the SU, i.e., $\Aoi_i(t)$. We will call $\Aoi_i(t)$ the Age of Information (AoI) of channel $i$. Since we assume the transition probability $q_i$ to be less than 1/2, the best estimate of PU occupancy at time-slot $t$ is simply the last observed channel occupancy:
$
\hatX_i(t) = X_i(t-\Aoi_i(t)).
$

Given this estimate $\hatX_i(t)$, the AoI $\Aoi_i(t)$, and the PU occupancy transition matrix $Q_i$, we can easily compute $\mathbb{P}\big(X_i(t)| \hatX_i(t)\big)$. Thus, the state of our MDP will be the following two-dimensional tuple $\big\{\hatX_i(t), \Aoi_i(t)\big\}.$
Here, $\hatX_i(t) \in \{0,1\}$ while $\Aoi_i(t) \in \mathbb{Z}^{+}$.

\paragraph*{\textbf{Actions}} In each time-slot $t$, the SU has only one decision to make - whether to access channel $i$ (set $u_i(t) = 1$) and pay an effective transmission cost $D$ in the hope of throughput reward, or not transmit (set $u_i(t) = 0$).

\paragraph*{\textbf{Rewards}} The reward in our MDP is straightforward. In each time-slot, the reward for channel $i$ equals throughput minus an effective transmission cost, i.e., $R_i(t) - D u_i(t)$. More importantly, we are interested in time-average reward, rather than discounted reward, which will influence the structure of the Bellman equations. 

\paragraph*{\textbf{Transitions}} Given our states, actions, and rewards, the following lemma describes how the MDP evolves over time.
\begin{lem} \label{lem:MDP Evolution}
    Given the state $\{\hatX_i(t), \Aoi_i(t)\}$ at time-slot $t$, the actual PU occupancy at time-slot $t$ is distributed as follows
    \begin{align}
    \P\big(X_i(t) = \hatX_i(t) \big)= \left[Q_i^{\Aoi_i(t)}\right]_{00},\\
    \P\big(X_i(t) \neq \hatX_i(t) \big) = \left[Q_i^{\Aoi_i(t)}\right]_{10}.
    \end{align}
    Here $[Q_i^\Aoi]_{mn}$ represents the $(m,n)$-the entry of $Q_i^\Aoi$, which is the $\Aoi$-th power of the matrix $Q_i$.
    Using this, we can now write down state transitions as a function of the channel access decision $u_i(t)$.
    \begin{align} \label{eq:Aoi update rule}
\Aoi_i(t+1) = 
    \begin{cases}
    \Aoi_i(t)+1, &\text{if } u_i(t)=0 \\
    1, &\text{if } u_i(t)=1.
    \end{cases}
\end{align}
\begin{align} \label{eq:estimate update rule}
\hatX_i(t+1) = 
    \begin{cases}
    \hatX_i(t), &\text{if } u_i(t)=0 \\
    X_i(t), &\text{if } u_i(t)=1.
    \end{cases}
\end{align}
\end{lem}

\begin{proof}
    The proof follows by expanding the multi-step Markov  transition probabilities and then leveraging the symmetry of the matrix $Q_i$. For more details, see Appendix B. 
\end{proof}

As a consequence of Lemma~\ref{lem:MDP Evolution}, the estimated PU channel occupancy at any time can be computed as just a function of the last known occupancy, the transition matrix $Q_i$ and the AoI. 
Given the MDP described above, we will now describe the Bellman equations for our decoupled single-channel problem. 
The Bellman equations characterize the set of conditions that an optimal policy must satisfy \cite{bertsekas2012dynamic}. 

Since the analytical approach is the same across all $N$ instances of the decoupled problem \eqref{eq:P1-d}, we will omit the channel identifier $i$ from our analysis below. To set up the Bellman equations, we introduce some more notation. Let $\lambda$ denote the time-average reward under the optimal policy. Further, let $S(\hatX(t),\Aoi(t))$ represent the differential value function, which measures the relative value of being in state $(\hatX(t),\Aoi(t))$ compared to this time-average reward. 

\begin{lem} \label{lem:Bellman equations}
The Bellman Equations for the decoupled single-channel problem are given by the following two equations.
\begin{align}\label{eq:Bellman Case 1}
\begin{aligned}
S(1,\Aoi) = \max\limits_{u \in \{0,1\}} \bigg\{& S(1,\Aoi+1), \big[Q^\Aoi\big]_{10}(S(0,1)+1)
\\ + \big[Q^\Aoi\big]_{00}&S(1,1) - D \bigg\}-\lambda, ~\forall \Aoi \in \mathbb{Z^+}.
\end{aligned}
\end{align}

\begin{align}\label{eq:Bellman Case 2}
\begin{aligned}
S(0,\Delta) = \max\limits_{u \in \{0,1\}} \bigg\{& S(0,\Aoi+1), \big[Q^\Aoi\big]_{00}(S(0,1)+1) \\+ \big[Q^\Aoi\big]_{10}&S(1,1) -D\bigg\} - \lambda , ~\forall \Aoi \in \mathbb{Z^+}.
\end{aligned}
\end{align}
\end{lem}

\begin{proof}
Combining Lemma \ref{lem:MDP Evolution} with evaluation of the Bellman recursion under the two possible PU occupancy states $X(t) \in \{0,1\}$ leads to the two equations above. For more details, see Appendix C. 
\end{proof}

An optimal policy must satisfy the system of equations \eqref{eq:Bellman Case 1} and \eqref{eq:Bellman Case 2} above. Before we find such a policy, we will look at policies with a threshold type structure. 
\begin{defn}
    A \textbf{threshold policy} with  the threshold $H$ for the decoupled single-channel  problem computes the spectrum access decisions for the SU via the following mapping.
    \begin{equation}
    \label{eq:threshold-policies}
        u(t) = \begin{cases}
        1, \text{ if } ~\hatX(t)= 0\\
        1, \text{ if } ~\hatX(t) = 1 \text{ and }\Aoi(t) \geq H\\
        0, \text{ otherwise}.
        \end{cases}
    \end{equation}
\end{defn}

Intuitively, a threshold policy of the form above allows the SU to keep utilizing the channel when it is  unoccupied by the PU ($\hatX(t)=0)$. As soon as it realizes that the PU is actually occupying the channel $(\hatX(t)=1)$, it stops transmitting and then waits for $H$ time-slots before trying a re-transmission. 
The following lemma characterizes the time-average reward for such threshold policies.
\begin{lem} \label{lem:Average Reward λ}
The time-average reward under the threshold policy is a function of the threshold $H$ and the effective transmission cost $D$.
\begin{align} \label{eq:Time-Average Reward λ}
\lambda(H,D)= \frac{\left[Q^{H}\right]_{10}-(\left[Q^{H}\right]_{10}+\left[ Q \right]_{10})D}{\left[ Q^{H} \right]_{10}+H\left[ Q \right]_{10}}.
\end{align}
\end{lem}

\begin{proof}
    To prove this result, we solve the Bellman equations from Lemma \ref{lem:Bellman equations} under a threshold policy with the parameter $H$. For more details, see Appendix D. 
\end{proof}



The structure of the threshold policies is motivated by our Markov channel occupancy assumption. If the SU attempted a transmission in the last time-slot and succeeded, we know that the PU had left the channel unoccupied. Since $q < 1/2$, the channel is still more likely to be unoccupied in the next time-slot, and our certainty about the channel being empty only decreases with time. So, the SU should keep transmitting until at some point in the future, when the PU comes back to this channel and causes a collision. At that point, we expect that at least for the next few time-slots, the PU will continue occupying the channel. So, the SU should wait for some threshold $H$ before it tries re-transmission. It turns out that simply choosing the best possible threshold over this simple class of policies leads to a policy that satisfies the Bellman equations and solves \eqref{eq:P1-d} optimally. 

\vspace{2mm}
\begin{mdframed}
[linewidth=1pt,linecolor=black]
\begin{thm} \label{thm:Optimal threshold H}
Given $q \in (0, 1/2)$, a threshold policy of the form \eqref{eq:threshold-policies} with threshold $H^*$ optimally solves the decoupled single-channel problem and satisfies the Bellman equations \eqref{eq:Bellman Case 1} and \eqref{eq:Bellman Case 2}. The optimal threshold $H^*$ satisfies
\begin{align} \label{eq:Optimal threshold condition}
&g(H^*) \leq \lambda \leq g(H^*-1),
\end{align}
where $g(H)\triangleq([Q^{H+1}]_{10}-[Q^H]_{10})\big(S(0,1)+1\big)$. If no such $H^*$ can be found, the optimal policy for the SU is to never transmit.
\end{thm}
\end{mdframed}

\begin{proof}
     To prove this result, we show that under the condition \eqref{eq:Optimal threshold condition}, the Bellman equations from Lemma \ref{lem:Bellman equations} are satisfied. This requires us to leverage the concavity of the term $[Q^H]_{10}$ as a function of the threshold $H$. We also show that at least one such optimal threshold $H^*$ always exists, although it does not need to be finite to derive a necessary condition for optimality involving $g(H)$. For more details, see Appendix E.
\end{proof}

A key insight is that proving the monotonicity of $g(H)$ is essential for ensuring the existence and uniqueness of the optimal threshold $H^*$, which in turn defines a stationary threshold structure. This means that the decision to transmit is based solely on the last observed PU occupancy $\hatX$ and the AoI $\Aoi$, providing a simple yet optimal strategy. Furthermore, if $H^*$ satisfies \eqref{eq:Optimal threshold condition}, then this threshold policy's differential value function $S(\hatX,\Aoi)$ satisfies the optimal Bellman recursion. Accordingly, the threshold policy with $H^*$ is optimal over all possible policies.
It turns out that, $g(H)$ is an increasing and concave function for all $q\in(0,1/2)$. Therefore, the optimal threshold $H^*$ can be efficiently found using the bisection method as stated in Theorem \ref{thm:Optimal threshold H}. 

\begin{figure}[h] 
     \centering
     \includegraphics[width=0.8\linewidth]{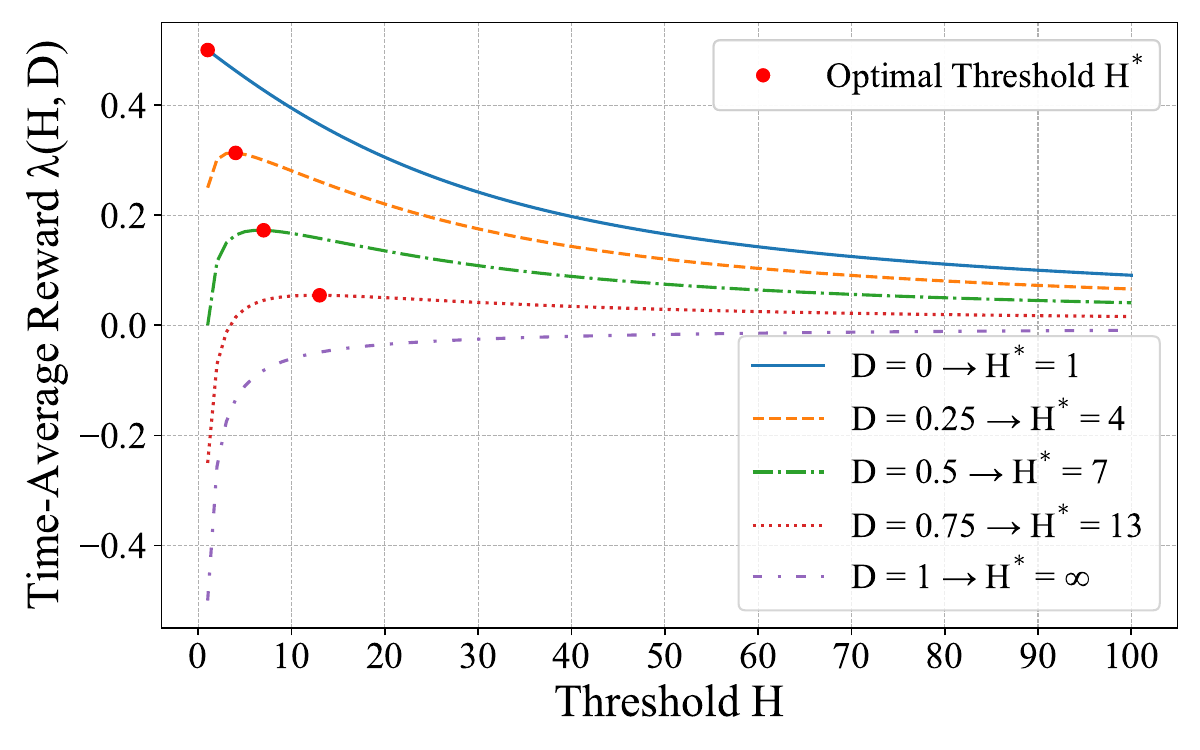}
         \vspace{-3mm}
     \caption{Time-average reward $\lambda(\Aoi,D)$ vs AoI $\Aoi$ with varying transmission cost $D=\{0, 0.25, 0.5, 0.75, 1\}$ and fixed transition probability $q=0.05$, highlighting the optimal threshold $H^*$.}
     \label{fig:λ(H,D) vs H}     
         \vspace{-3mm}
\end{figure}

Fig. \ref{fig:λ(H,D) vs H} illustrates the time-average reward $\lambda(H,D)$ as a function of the threshold $H$ for various transmission costs $D$. The transition probability is fixed at $q=0.05$. As the effective transmissions cost $D$ increases, the optimal threshold $H^*$ also increases. This means that the SU should wait longer before trying a re-transmission, as frequent attempts lead to more collisions. When $D=0$, the optimal threshold is $H^*=1$, indicating that updates occur as frequently as possible since there is no penalty for transmitting. Conversely, when $D=1$, the optimal threshold becomes $H^*=\infty$, which implies that the SU never tries to transmit, as the effective transmission cost always cancels out immediate reward gained from successfully transmitting a single packet.

\section{Spectrum Access Policies} \label{sec:Scheduling Policies}
In the previous section, we established that to solve the decoupled single-channel problem, an optimal policy has the threshold type structure as described by \eqref{eq:threshold-policies}. This solution can be applied to each channel $i$ to compute an optimal threshold $H^{*}_i$, thus solving the Lagrangian relaxed problem. However, it still does not provide us a solution for the original problem \eqref{eq:P1, Unified} with strict spectrum access constraints. To do so, we will introduce the Whittle index approach.



The foundation of index-based decision-making in stochastic control traces back to the ``Gittins Index'' which was introduced by John C. Gittins in 1979 to solve the multi-armed bandit (MAB) problem and dynamic allocation indices \cite{gittins1979bandit}. The Gittins Index provides an near-optimal solution for classical multi-armed bandits with reward function and probability of termination, offering an efficient way to prioritize arms without solving the full DP problem. 
Building on this concept, Peter Whittle extended the notion to restless multi-armed bandits (RMABs) in 1988,  where the states of non-played arms can also evolve over time \cite{whittle1988restless}. In his seminal work, he introduced the ``Whittle Index'' which allows for tractable approximations of the optimal scheduling policy in resource allocation problems. Note that the ``restless'' part is crucial to modeling our spectrum access problem, as the PU occupancy is dynamic, i.e., restless.

\subsection{Whittle Index Policy} \label{subsec:Whittle Index}
To apply the Whittle index approach to our problem and find an optimal policy for $\eqref{eq:P1, Unified}$, we first need to verify a key structural property known as \textit{indexability} for the decoupled problems, which ensures that a well-defined index can be assigned to each state. 


\begin{defn} \textbf{Indexability} \cite{whittle1988restless}.
Given cost $D$, let $\mathscr{P}(D)$ be the set of states for which the optimal action in the decoupled problem is not to transmit. The problem is \textit{indexable} if the set $\mathscr{P}(D)$ monotonically increases from the empty set to the entire state space, as $D$ increases from $0$ to $\infty$.    
\end{defn}

\begin{defn}  \label{def: Whittle Index Policy}
\textbf{Whittle index} \cite{whittle1988restless}. Consider the decoupled problem and denote by $W(\hatX, \Aoi)$ the Whittle index in state $(\hatX, \Aoi)$. Given indexability, $W(\hatX, \Aoi)$ is the minimum cost $D$ that makes both actions ($u=1$ and $u=0$) equally desirable in state $(\hatX, \Aoi)$. 
\end{defn}

The following lemma establishes that indexability holds for the decoupled single-channel problems.
\begin{lem} \label{lem:Indexability}
The \textit{indexability} property holds for each decoupled single-channel problem \eqref{eq:P1-d}, given $q_i \in (0, 1/2)$.
\end{lem}

\begin{proof}
    Since the optimal threshold $H^*(D)$ increases monotonically with the effective transmission cost $D$, 
    the decoupled problem is indexable. For more details, see Appendix F. 
\end{proof}

Since the \textit{indexability} condition holds for each decoupled single-channel problem, we can proceed to compute the corresponding Whittle index. The Whittle index is derived by identifying the cost at which both actions ($u=1$ and $u=0$) are equally desirable for a given state. The following theorem provides a closed-form expression for the Whittle index for spectrum sharing.
\begin{mdframed}
[linewidth=1pt,linecolor=black]
\begin{thm} \label{thm:Whittle Index}
For the $i$-th decoupled problem \eqref{eq:P1-d}, the Whittle index $W_i(\hatX_i, \Aoi_i)$ is given by \\
$W_i(\hatX_i, \Aoi_i)$
\begin{align} \label{eq:Whittle index}
\triangleq \begin{cases}
    \infty, \hspace{56.5mm} \text{if } \hatX_i = 0\\
    \frac{\Aoi_i([Q_i^{\Aoi_i-1}]_{10}-[Q_i^{\Aoi_i}]_{10})+[Q_i^{\Aoi_i}]_{10}}{(\Aoi_i-1)([Q_i^{\Aoi_i-1}]_{10}-[Q_i^{\Aoi_i}]_{10})+[Q_i^{\Aoi_i}]_{10}+[Q_i]_{10}}, \enspace \text{if } \hatX_i=1.
\end{cases}
\end{align}
\end{thm}
\end{mdframed}

\begin{proof}
    Using the indexability established in Lemma \ref{lem:Indexability}, the Whittle index is defined as the cost that makes both actions equally desirable at a given state. So to find the index, the transmission cost $D$ needs to be set such that this condition is satisfied. For more details, see Appendix F. 
\end{proof}

At a high level, the Whittle index of a channel at state $(\hatX_i(t), \Aoi_i(t))$ measures how valuable it is to schedule a transmission on that channel, given that the last observed PU occupancy was $\hatX_i(t)$ and this observation was made $\Aoi_i(t)$ time-slots ago. Clearly, if we know that the PU is not occupying the channel ($\hatX_i(t)=0$), we should attempt a transmission, which suggests the infinite Whittle index value. However, if we last observed the channel to be occupied by the PU ($\hatX_i(t)=1$), the Whittle index is low for small values of AoI and gradually increases as $\Aoi_i(t)$ increases, since it becomes more and more likely that the channel occupancy flipped. Using \eqref{eq:Whittle index}, we can now design the Whittle index policy in Algorithm~\ref{alg:Whittle Index + Threshold Policy} to address the original SU throughput maximization problem under spectrum access constraints and collision penalties \eqref{eq:P1, Unified}. 


\begin{algorithm}[h!]     
\caption{Whittle Index for Spectrum Sharing}  \label{alg:Whittle Index + Threshold Policy}  

\KwIn{maximum number of SU channels $L$, collision penalty $\gamma$, and time horizon $T$.}
\KwOut{Spectrum access policy $u_i(t), \forall i \in [N], \forall t \in [T].$}

\BlankLine
\textbf{Initialization:} \\

$\Aoi_{i}(0) \gets 1,~\forall i \in [N]$ 

$\hatX_i(0) \gets 1, \forall i \in [N]$ 

\BlankLine
\For{$t = 0$ \KwTo $T$}{
    Sort channels by Whittle index $W_i(\hatX_i(t), \Aoi_i(t))$

    $\text{For the top } L \text{ channels: }$ 
    
    $u_i(t) \gets   \begin{cases} 
        1, \text{if } \Aoi_i(t) \geq H^*_i(q_i,\frac{\gamma}{1+\gamma})\\
        0, \text{otherwise}
    \end{cases}$

    $\text{For the rest: } u_i(t) \gets 0$

    \For{$i = 1$ \KwTo $N$}{
    $\Aoi_i(t+1) \gets \begin{cases}
        \Aoi_i(t)+1, & \text{ if }u_i(t)=0\\
        1,& \text{ if }u_i(t)=1
    \end{cases}$
    
    $\hatX_i(t+1) \gets \begin{cases}
        \hatX_i(t), & \text{ if }u_i(t) = 0\\
        X_i(t), & \text{ if }u_i(t)=1
    \end{cases}$
    }     
}
\end{algorithm}
Intuitively, the algorithm prioritizes channels with a higher likelihood of being free immediately and for longer periods by assigning them larger Whittle index values, which reflect the minimum cost (or ``price'') an SU is willing to pay, indicating the perceived worth of channel access. 
The algorithm selects the top $L$ channels based on the Whittle Index and then chooses those which exceed their optimal threshold under the collision penalty. 
The AoI and last observed PU occupancy states are then updated accordingly for each channel based on the transmission decision. For tie-breaking among channels with the same Whittle index values, we prioritize ones with the lower AoI, since it indicates more certainty about channel occupancy information.

\section{Correlated Channel Occupancies}\label{sec:Correlated Channel}
A drawback of our analysis and the Whittle index approach is that it does not immediately extend to settings where the PU occupancy is correlated across channels. A typical example would be the case where the PU uses a contiguous band of $B$ channels, and the center of this band moves randomly across the available spectrum over time. To tackle such settings, we first develop a heuristic index function motivated by our intuitive understanding of the Whittle index. This heuristic index policy will then provide us a pathway to solving the general scenario.

\subsection{Heuristic Index Policy}
Let's suppose that channel $i$ was last known to be occupied by the PU, i.e., $\hatX_i(t) = 1$. Further, let's suppose that this observation is $\Aoi_i(t)$ time-slots old. The following lemma characterizes how many packets the SU would be able to transmit if it chooses channel $i$ in the next time-slot and continues transmitting until it first observes a collision with the PU. We denote this quantity as $G_i(t)$.
\begin{lem} \label{lem:Approximated Whittle index}
Given channel $i$ occupancy information $\hatX_i(t) = 1$ and the AoI $\Aoi_i(t)$ at time-slot $t$, the number of packets $G_i(t)$ that the SU can transmit on channel $i$ starting at time-slot $t$ until it faces its first collision satisfies the following
\begin{equation}
    \mathbb{E}[G_i(t)] = \frac{[Q_i^{\Aoi_i}]_{10}}{[Q_i]_{10}}.
\end{equation}
\begin{proof}
Intuitively, $[Q_i^{\Aoi_i}]_{10}$ denotes the probability that the channel transitions from occupied to free in the next time-slot, conditioned on $\Aoi_i$-old occupancy information. The term $1/[Q_i]_{10}$ corresponds to the expected free duration from the moment the channel just became free until the first collision. Therefore, the product of these two terms estimates the expected number of packets that can be transmitted before the first collision. For more details, see Appendix G. 
\end{proof}
\end{lem}

We can think of this expectation as the potential reward for choosing channel $i$ at time-slot $t$. This motivates a heuristic index policy for our spectrum access problem - where we use the following index function $V_i(\hatX_i, \Aoi_i)$ instead of the Whittle index functions $W_i(\hatX_i,\Aoi_i)$ derived in Theorem~\ref{thm:Whittle Index}.
\begin{align} \label{eq:Heuristic index}
V_i(\hatX_i, \Aoi_i) \triangleq \begin{cases}
    \infty&, \text{ if } \hatX_i = 0\\
    \frac{[Q_i^{\Aoi_i}]_{10}}{[Q_i]_{10}}&, \text{ if } \hatX_i =1.
\end{cases}
\end{align}

By substituting \eqref{eq:Heuristic index} in place of \eqref{eq:Whittle index} in Algorithm \ref{alg:Whittle Index + Threshold Policy}, we can implement a heuristic index policy, which leverages the simplified index for efficient spectrum access decisions. We will show later via simulations in Section~\ref{sec:Numerical results} that this heuristic index performs almost as well as the Whittle index for independent PU occupancies across the channels. This heuristic index in \eqref{eq:Heuristic index} can be understood as an intuitive simplification of the Whittle index in \eqref{eq:Whittle index}. Through this simplification, we are able to interpret the key quantities that drives decision making for a good spectrum access policy. 

The benefits of  this simpler index are evident in when the PU occupancies are correlated across channels. The new index can be reformulated as a function of the observed PU occupancy history in correlated scenarios (since observations of channel $j$ will affect estimates for channel $i$ even if $i\neq j$).
\begin{align} \label{eq:Correlated index}
&V_i(\text{PU occupancy history}) \notag \\
&\triangleq \begin{cases}
    \infty,& \text{ if } \hatX_i = 0\\
    \mathbb{E}[G_i(t) | \text{PU occupancy history}],& \text{ if } \hatX_i=1.
\end{cases}
\end{align}

This conditional expectation is often easy to compute \textit{or estimate using data}, and we can utilize the resulting index functions $V_i$ to implement an index policy for correlated or data-driven scenarios. In Section~\ref{sec:Numerical results} we will show that we can effectively do this for PU occupancies that move in a contiguous block over time. The heuristic index significantly outperforms all baseline policies, thus extending our results and intuition to the correlated scenario. 


\section{Numerical Results} \label{sec:Numerical results}

In this section, we present simulation results that validate our analytical findings on the performance of the Whittle index. The simulations are conducted with a time horizon of $T=30000$, and each data point represents an average over 100 independent simulation runs. Unless specified otherwise, the Markov channels are set to have transition probabilities evenly spaced from $q_{min}=0.1$ to $q_{max} = 0.5$. 
Additionally, the collision penalty $\gamma$ is set to a default value of 0.5.

\subsection{Independent Channel Occupancies}
Under the independent channel setting, we compare the performance of four different policies: 1) The Pure Random policy, that selects $L$ channels uniformly at random regardless of previous selections, which serves as a lower bound for performance and is similar to the random access characteristics of the slotted ALOHA protocol \cite{abramson1970aloha, abramson1973packet, kleinrock1973packet}. 2) The Check Empty + Random policy which continues transmissions on channels that were found to be empty in the previous time-slot, and chooses a new channel uniformly at random for each scheduling decision that leads to a collision. This approach is similar to Carrier Sense Multiple Access (CSMA) protocol \cite{kleinrock1975packet_part1}, where a node senses the channel and transmits only if it is free, and making a random decision otherwise. 3) The Whittle index policy, described in Section \ref{sec:Scheduling Policies}, which selects channels with the highest Whittle indices when collision occurs. 4) The Heuristic index policy described in Section \ref{sec:Correlated Channel}, which utilizes \eqref{eq:Heuristic index} to approximate the Whittle index. 

Figure \ref{fig:Varying_L} depicts the impact of the selected number of channels $L$ on system performance. From Figure \ref{fig:Varying_L_vs_Throughput}, the average SU throughput normalized per channel decreases as $L$ increases due to higher competition, with the index policies performing best and the Pure Random policy worst. From Figure \ref{fig:Varying_L_vs_Loss}, the average collision rate rises with $L$, but the Whittle index policy reduces collisions effectively. Notably, when selecting between $1$ channel and $N/4$ channels, our Whittle index policy achieves up to a \textbf{19\% average SU throughput gain} and a \textbf{39\% collision rate reduction} compared to the Check Empty + Random policy. The Heuristic index policy also closely approximates the performance of the Whittle index.

\begin{figure}[h!]
    \centering
    \subfigure[Normalized throughput vs. $L$]{
        \includegraphics[width=0.47\columnwidth]{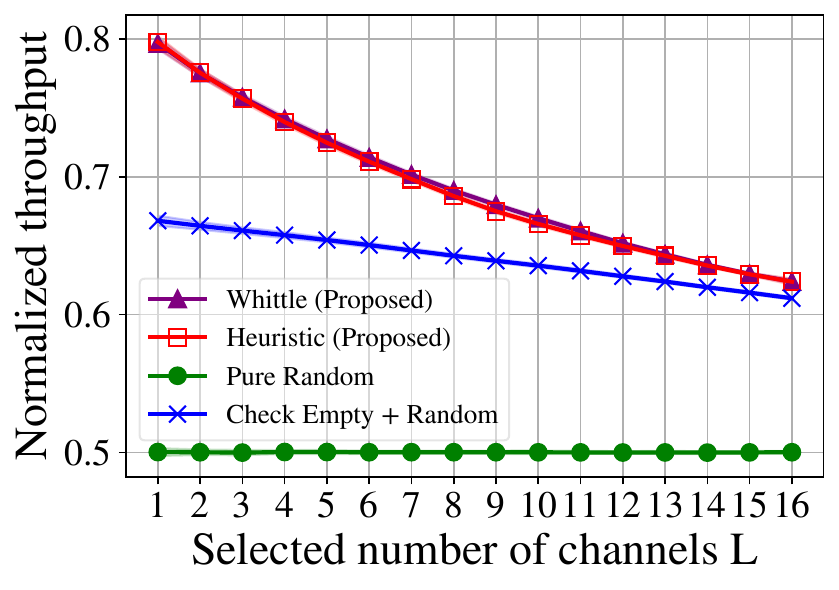}
        \label{fig:Varying_L_vs_Throughput}
    }   
    \hspace{-10pt}
    \subfigure[Average collision rate vs. $L$]{
        \includegraphics[width=0.47\columnwidth]{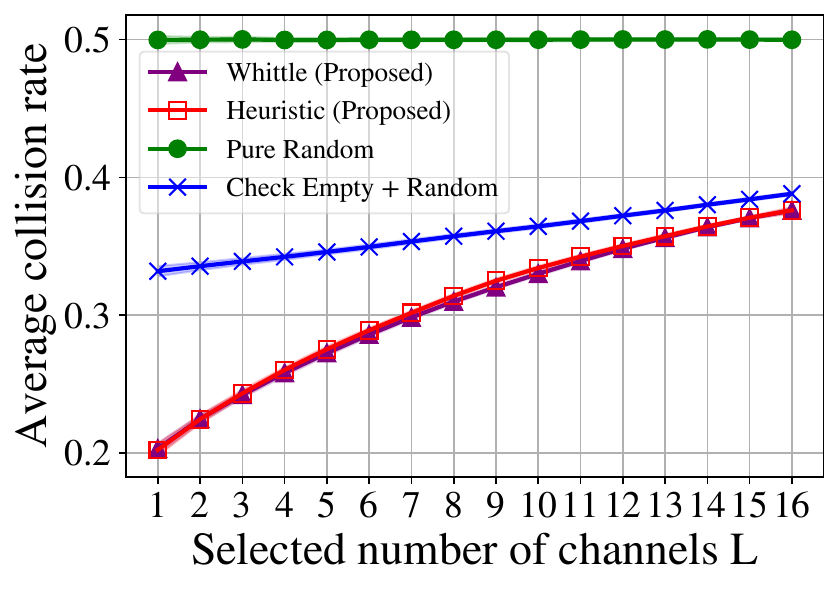}
        \label{fig:Varying_L_vs_Loss}
    }
    \caption{Performances of different scheduling policies vs. selected channel number $L$ under independent channel occupancies (fixed total channel number $N=32$).}
    \label{fig:Varying_L}
    \vspace{-3mm}
\end{figure}

\begin{figure}[hbt!]
    \centering
    \subfigure[Normalized throughput vs. $q_{max}$]{
        \includegraphics[width=0.47\columnwidth]{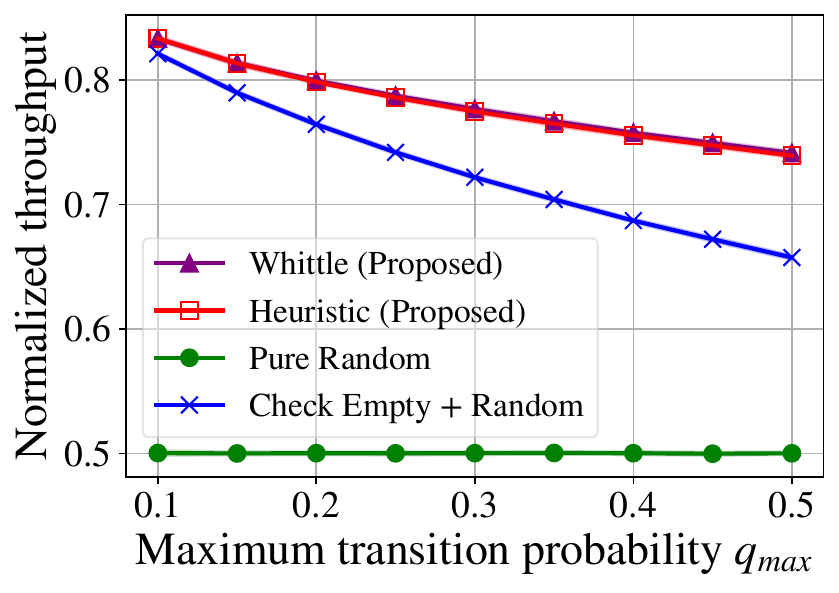}
        \label{fig:Varying_q_max_Throughput}
    }
    \hspace{-10 pt}
    \subfigure[Average collision rate vs. $q_{max}$]{
        \includegraphics[width=0.47\columnwidth]{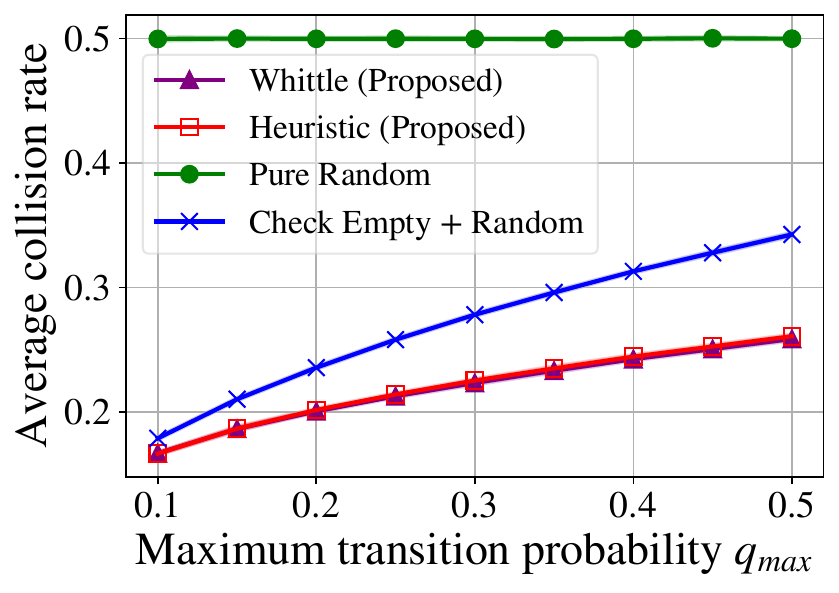}
        \label{fig:Varying_q_max_Loss}
    }
    \caption{Performances of different scheduling policies vs. maximum transition probability $q_{max}$ under independent channel occupancies (fixed total channel number $N=32$ and selected channel number $L=4$).}
    \label{fig:Varying_q_max}
    \vspace{-3mm}
\end{figure}

Next, Figure \ref{fig:Varying_q_max} illustrates how changes in $q_{max} \leq 0.5$ influence system performance. Here, $q_{max}$ 
which quantifies the uncertainty in channel occupancy. From Figure \ref{fig:Varying_q_max_Throughput}, the average SU throughput decreases as $q_{max}$ increases, indicating that higher channel uncertainty negatively affects SU throughput. The Whittle index policy consistently outperforms the Check Empty + Random policy, while the Pure Random policy remains the worst.  
The Whittle index policy achieves up to a \textbf{13\% gain in average SU throughput} and consistently maintains a \textbf{24\% reduction in the collision rate} compared to the Check Empty + Random policy, regardless of $q_{max}$, demonstrating its robustness against channel uncertainty. 

Finally, to understand the effect of system scale, Figure \ref{fig:Varying_N} presents the results under varying total channels $N$ with a fixed selected number $L=N/4$. The overall performance remains stable across different values $N$ since the proportion of selected channels stays constant. As before, the Whittle index policy consistently achieves higher throughput and lower collision loss compared to other strategies. Initially, performance variance is higher due to the small number of stable channels but stabilizes as $N$ increases.  In Figure \ref{fig:Varying_N_Throughput} and \ref{fig:Varying_N_Loss}, the Whittle index policy achieves the highest throughput and the lowest collision rate, followed by the Check Empty + Random policy. While the Pure Random policy consistently performs the worst in both aspects across all values of $N$. 
Figure \ref{fig:Varying_N_overall} highlights these trends by comparing SU throughput and collision rate across different policies. This figure serves as a performance metric where the upper-left region indicates superior performance, while the lower-right region represents inferior performance. Our Whittle and heuristic index policies consistently maintains superior performance, achieving high throughput while minimizing collision loss. Crucially, the Whittle and heuristic index strategies remain effective regardless of system scale, demonstrating their robustness in maintaining performance even as the total number of available channels increases. 

\begin{figure}[h]
    \vspace{-2.0mm}
    \centering
    \subfigure[Normalized throughput vs. $N$]{
        \includegraphics[width=0.47\columnwidth]{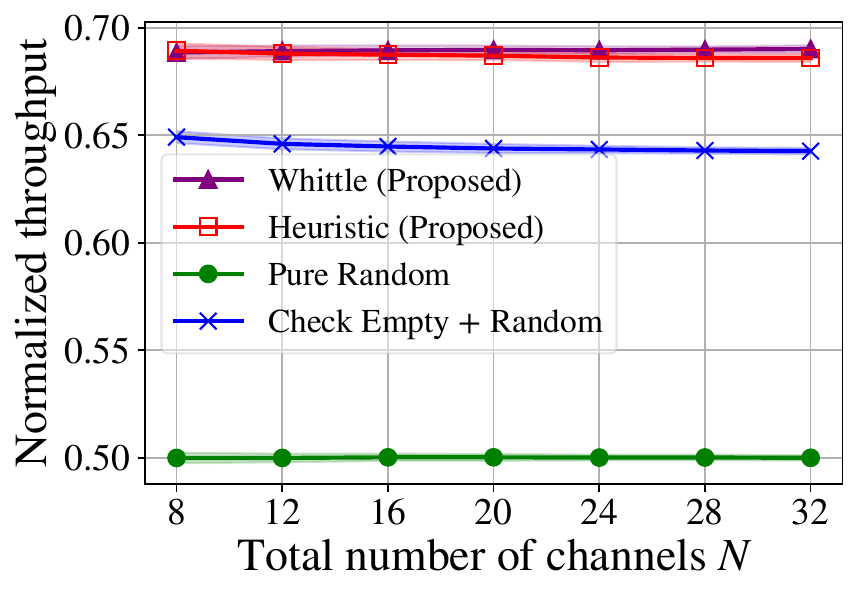}
         \label{fig:Varying_N_Throughput}
    }
    \hspace{-10 pt}
    \subfigure[Average collision rate vs. $N$]{
        \includegraphics[width=0.47\columnwidth]{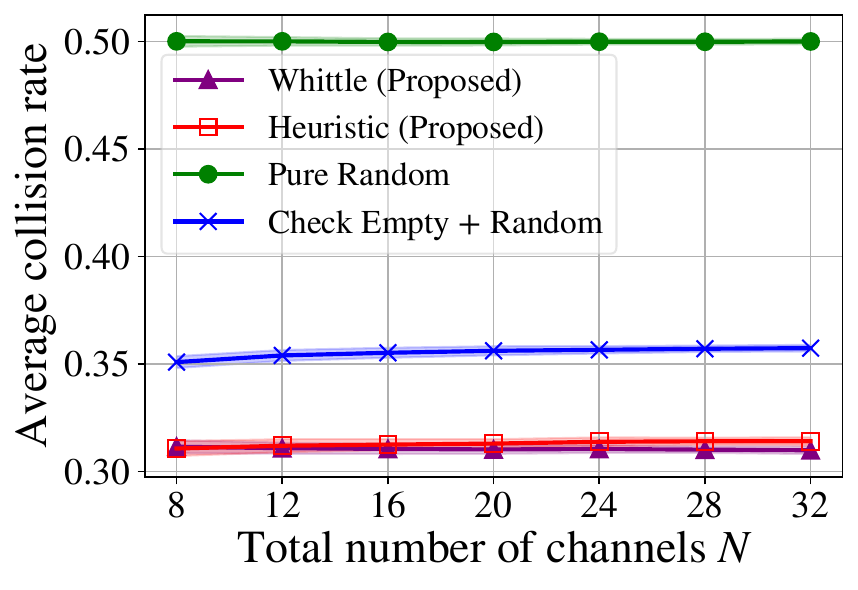}
         \label{fig:Varying_N_Loss}
    }
    \hfill
    \subfigure[Average collision rate vs. Normalized throughput for varying $N$]{
         \includegraphics[width=0.85\columnwidth]{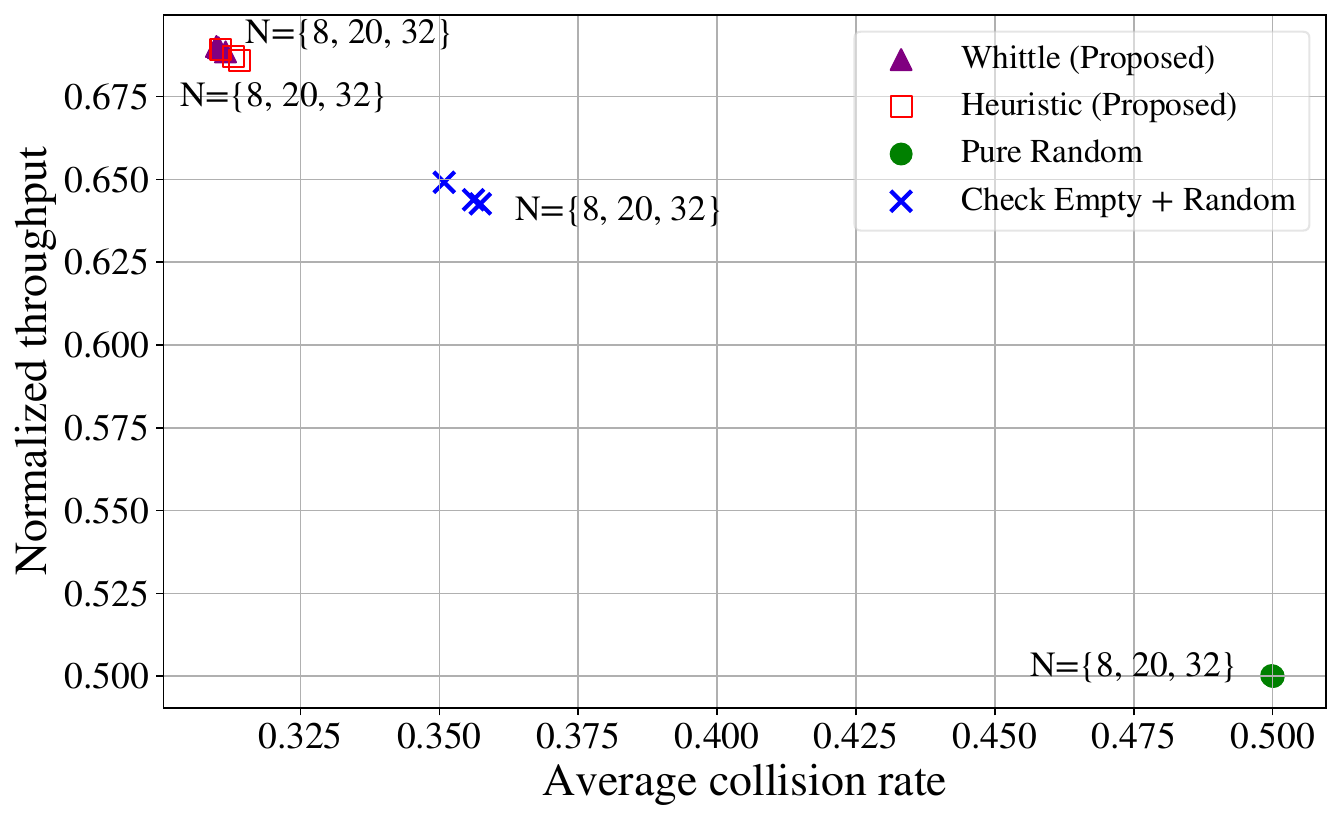}
          \label{fig:Varying_N_overall}
    }
    
    \vspace{-2mm}
    
    \caption{Performances of different scheduling policies vs. total channel number $N$ under independent channel occupancies (fixed selection ratio $L/N=1/4$).}
    \label{fig:Varying_N}
    
    \vspace{-1.5mm}
\end{figure}

\subsection{Learning with Stationary Channel Statistics} \label{subsec:Stationary learning}
We now repeat the numerical simulation from Figure \ref{fig:Varying_L} but with two variations of our proposed policies, that do not know the underlying Markov transition matrices $Q_i$ in advance. So along with the four algorithms already plotted, we add the Whittle index policy with learning, and the Heuristic index policy with learning. For these, the transition probabilities $q_i$ are not known a priori and instead estimated over time using transmission outcomes. In this subsection we assume the underlying transition probabilities remain fixed (stationary) over time.

The online estimate of the collision probabilities $\hat{q}_i(t)$ can be computed as a function of the number of transitions of the channel state ($0\rightarrow 0$ and $0 \rightarrow 1$) observed by the SU up to time $t$. The precise estimator is given by

\begin{align} \label{eq:estimator q_MLE}
 \hat{q}^{MLE}_i(t) = \frac{N^{0\rightarrow 1}_i(t)}{N^{0\rightarrow 1}_i(t) + N^{0\rightarrow 0}_i(t)},
\end{align}
where $N^{\hat{X}\rightarrow \hat{Y}}_i(t)$ gives the cumulative number of observed transitions of channel $i$'s occupancy from state $\hat{X}$ to state $\hat{Y}$ up to time $t$.
This estimator is the Maximum Likelihood Estimator (MLE) for the transition probability of a two-state Markov chain. It is derived by maximizing the likelihood of observing the transition counts $N^{0\rightarrow 1}_i(t)$ and $N^{0\rightarrow 0}_i(t)$. which simplifies to the ratio of observed event frequencies \cite{anderson1957statistical}.

\begin{figure}[hbt!]
    \vspace{-2mm}
    \centering
        
    \subfigure[Normalized throughput vs. $L$]{
        \includegraphics[width=0.9\columnwidth]{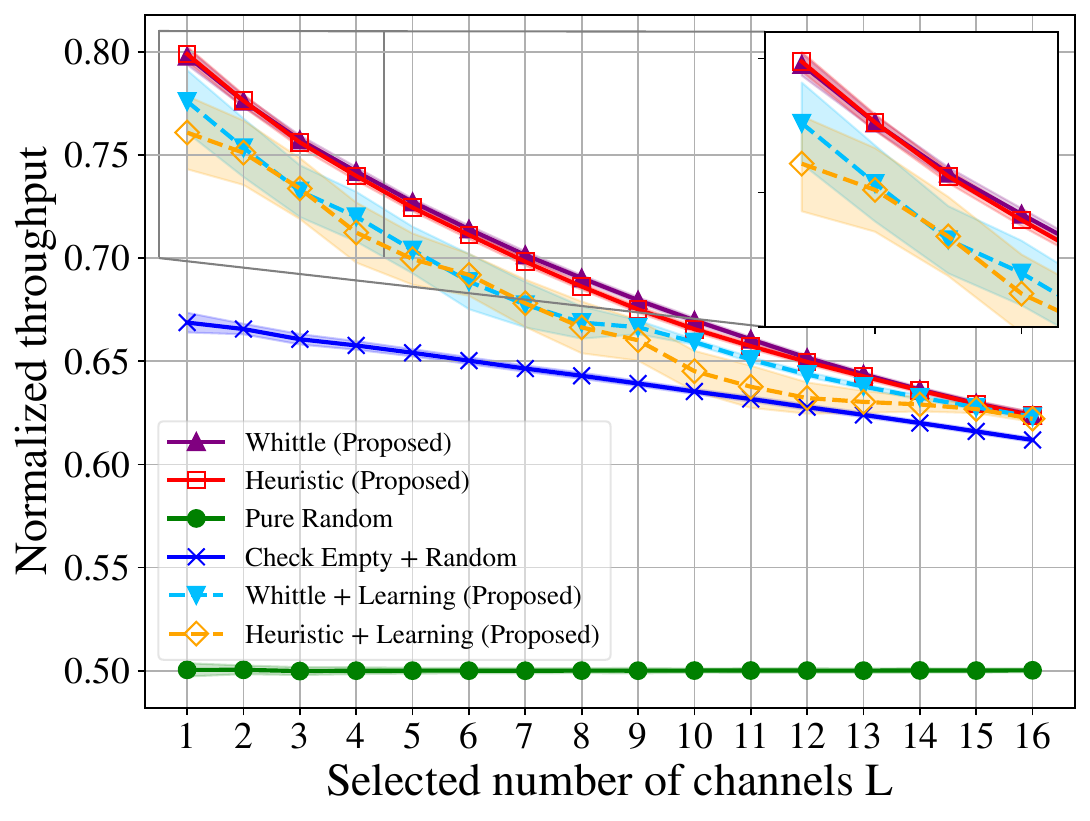}
        \label{fig:Stationary Learning:Varying_L_vs_Throughput}
    }
    
    \subfigure[Average collision rate vs. $L$]{
        \includegraphics[width=0.9\columnwidth]{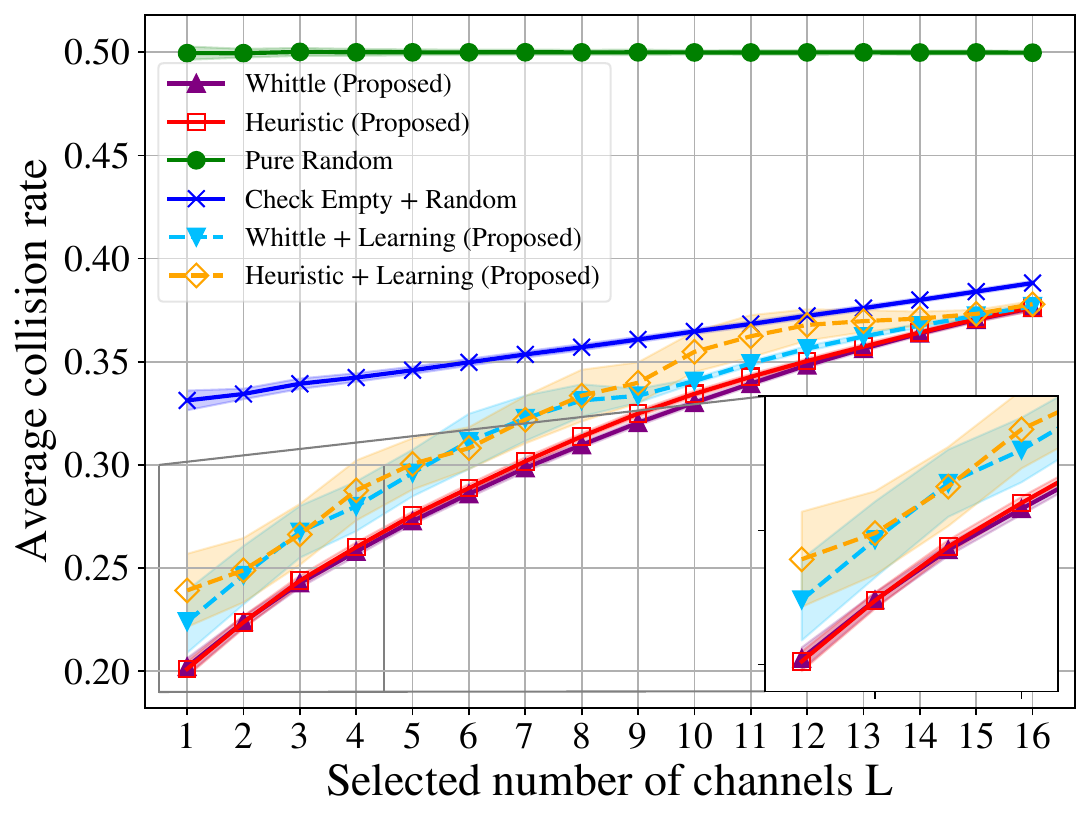}
        \label{fig:Stationary Learning:Varying_L_vs_Loss}
    }  
   
    \caption{Performances of different scheduling policies with and without online estimation of $\hat{q}^{MLE}_i$.}
    \label{fig:Stationary Learning:Varying_L}  
    \vspace{-2mm} 
\end{figure}

Figure \ref{fig:Stationary Learning:Varying_L} depicts the impact of the selected number of channel $L$ on system performance. The learning-based Whittle and heuristic policies maintain performance levels that closely approximate their non-learning counterparts. Notably, the Whittle + Learning policy achieves up to a \textbf{16\% average SU throughput gain} and a \textbf{32\% collision rate reduction} compare to the Check Empty + Random policy. The Heuristic + Learning policy also closely approximates the performance of the Whittle + Learning policy. The learning policies exhibit slightly higher standard deviation across individual runs compared to their non-learning counterparts, due to the online estimation of $q_i$ values.  

\subsection{Learning with Non-Stationary Channel Statistics}
We now relax the fixed-transition assumption of Subsection~\ref{subsec:Stationary learning} and consider non-stationary environment where the Markov transition probabilities $q_i(t)$ may evolve slowly over time due to changes in PU activity or channel conditions, which we model as piecewise linearly varying. Specifically, half of the channels exhibits increasing trends while the other half exhibit decreasing trends over time. In this setting, the standard MLE estimator in \eqref{eq:estimator q_MLE} is not well-suited as it gives equal weight to all past observations. As a consequence, the resulting estimate can lag behind the true dynamics, resulting in degraded performance. 

To address this challenge, we introduce a window-based adaptive estimator that employs an exponential forgetting mechanism~\cite{haykin2002adaptive}, which downweights older observations by applying a forgetting factor $\alpha\in(0,1]$ and emphasizes recent data. First, we update the exponentially weighted cumulative transition counts at each time slot $t$ according to the following recursive rule:

\begin{equation} \label{eq:recursive_forgetting_factor}
\tilde{N}^{\hat{X}\rightarrow \hat{Y}}_i(t)=
\begin{cases}
\begin{array}{l}
\lfloor \alpha \cdot \tilde{N}^{\hat{X}\rightarrow \hat{Y}}_i(t-1)\rfloor 
+ \Delta \tilde{N}^{\hat{X}\rightarrow \hat{Y}}_i(t), \\
\qquad \qquad \qquad \text{if } (t-1)=kT_w ~\forall k\in \mathbb{Z}^+,
\end{array} \\[1.5em]
\begin{array}{l}
\tilde{N}^{\hat{X}\rightarrow \hat{Y}}_i(t-1)
+ \Delta \tilde{N}^{\hat{X}\rightarrow \hat{Y}}_i(t), ~\text{otherwise},
\end{array}
\end{cases}
\end{equation}
where $\Delta \tilde{N}^{\hat{X}\rightarrow \hat{Y}}_i(t)$ is the incremental number of observed transitions of channel $i$ since the previous update, and $T_w$ denotes the observation window length in time slots. 

Then, using the updated counts, we compute the Exponentially Weighted Maximum Likelihood Estimator (EW-MLE) as 

\begin{equation} \label{eq:estimator q_EW-MLE}
 \hat{q}^{EW-MLE}_i(t) = \frac{\tilde{N}^{0\rightarrow 1}_i(t)}{\tilde{N}^{0\rightarrow 1}_i(t) + \tilde{N}^{0\rightarrow 0}_i(t)},
\end{equation}
which reduces to the standard MLE in \eqref{eq:estimator q_MLE} when $\alpha=1$. Within each window, the estimator behaves as the standard MLE by accumulating transition counts without decay, while the forgetting factor $\alpha$ is applied at the boundary between consecutive windows to downweight older data.

Figure \ref{fig:Non-Stationary Learning:Varying_L} depicts the performance of various scheduling policies under varying $L$, highlighting the impact of the exponentially weighted online estimator $\hat{q}^{EW-MLE}$. The learning-based Whittle and heuristic policies maintain performance levels that closely approximate their non-learning counterparts. Notably, the Whittle + Learning policy achieves up to a \textbf{9\% average SU throughput gain} and a \textbf{22\% collision rate reduction} compare to the Check Empty + Random policy. The Heuristic + Learning policy also closely approximates the performance of the Whittle + Learning policy. Even under non-stationary transition probabilities $q_i(t)$, our approach maintains high throughput and low collision rates, demonstrating its adaptability.  

\begin{figure}[hbt!]
    \vspace{-2mm}
    \centering
     
    \subfigure[Normalized throughput vs. $L$]{
        \includegraphics[width=0.9\columnwidth]{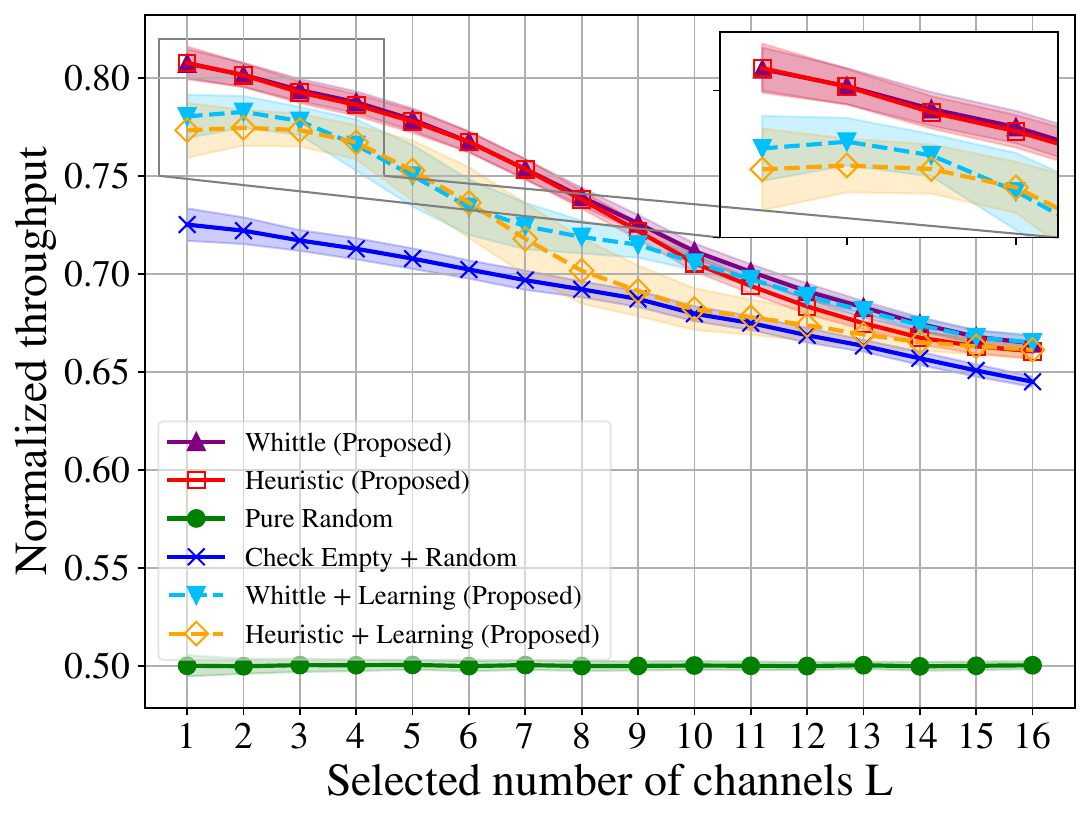}%
        \label{fig:Non-Stationary Learning:Varying_L_vs_Throughput}
    }
     
    \subfigure[Average collision rate vs. $L$]{
        \includegraphics[width=0.9\columnwidth]{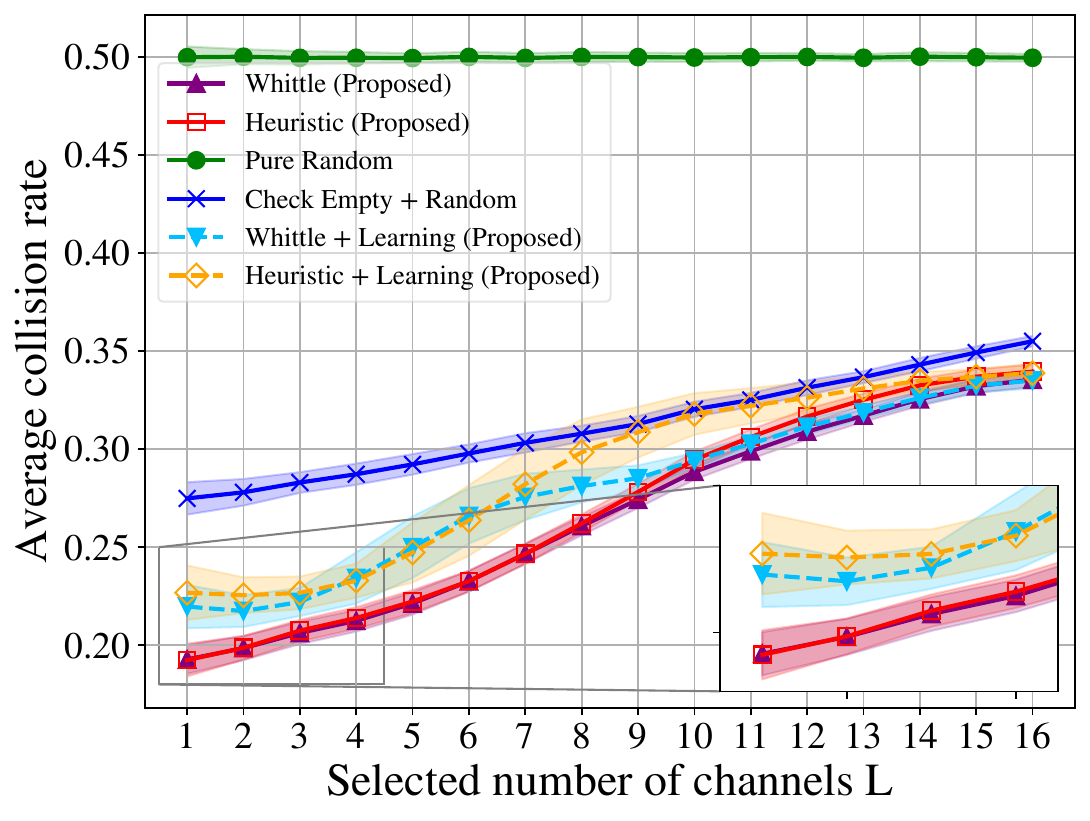}%
        \label{fig:Non-Stationary Learning:Varying_L_vs_Loss}
    }  
    
    \caption{Performances of different scheduling policies with and without online estimation of $\hat{q}^{EW-MLE}_i$.}
    \label{fig:Non-Stationary Learning:Varying_L}  
    \vspace{-2mm} 
\end{figure}

\subsection{Correlated Channel Occupancies}
We now study the correlated channel setting. Here, we compare the performance of three different policies: 1) The Pure Random policy, 2) The Check Empty + Random policy and 3) The Heuristic index policy described in Section \ref{sec:Correlated Channel}, which utilizes \eqref{eq:Correlated index} for computing heuristic indices from PU occupancy history. We assume that the PU occupies a contiguous band of $B=12$ channels out of a total of $N=16$ available channels, and the center position of this band follows a discretized Gaussian random walk over time, with noise variance $\sigma^2=1$.

To find an expression for the correlated index based on \eqref{eq:Correlated index}, we approximate the expected free duration of channel $i$ given all prior observation history. We use the idea of mean first-passage times (MFPT) in a Brownian motion process \cite{redner2001guide}, which calculates how long the process takes to hit a specified threshold and is known to scale quadratically with the distance between the current value and the threshold. The correlated index for channel $i$ in this case takes the following form
\begin{align}
&V_i(\text{PU occupancy history})= \notag \\
&\begin{cases}
\infty, \text{ if } \hatX_i = 0\\
\Big[ 1 - \Phi(\frac{i+B/2-\hatP(t)}{\sqrt{\bar\Aoi}}) + \Phi(\frac{i-B/2-\hatP(t)}{\sqrt{\bar\Aoi}})\Big]|i-\hatP(t)|^2, &\text{ o.w.} 
\end{cases}
\end{align}
where $\Phi(\cdot)$ is the standard normal CDF, $\bar{\Delta}$ is the current average AoI across all channels, and $\hatP(t)$ is the current estimate of the center of the PU occupancy band given past observations. We provide a detailed description of how we model correlated channel occupancies for this block setting and compute the corresponding heuristic index in Appendix H. 

Figure \ref{fig:Correlated_Varying_L} depicts the impact of the selected number of channels $L$ on system performance. From Figure \ref{fig:Correlated_Varying_L_Throughput}, the average SU throughput normalized per channel decreases as $L$ increases due to increased contention, with the heuristic index policy performing best. Given that there are $N=16$ total channels and $B=12$ of them are occupied by the PU, the SU can only succeed on the remaining 4 channels. Despite this constraint, the Heuristic index policy consistently achieves over 80\% success probability, while the Check Empty + Random policy remains below 80\%. This clearly shows the advantage of incorporating occupancy correlations: the Heuristic index policy guides the SU towards the available channels, even when the system is heavily loaded. From Figure \ref{fig:Correlated_Varying_L_Loss}, the average collision rate rises with $L$, but again, the Heuristic index policy effectively mitigates collisions by utilizing correlated PU behavior. In particular, compared to the Check Empty + Random policy, our heuristic index policy achieves up to a \textbf{43\% gain in average SU throughput} and an \textbf{82\% reduction in the collision rate} when the number of selected channels range from 1 to the maximum 4 available, significantly outperforming it in the most critical operational range, where scheduling is necessary. This underscores the value of leveraging PU correlation to minimize collisions in spectrum access.

\begin{figure}[h]
    \vspace{-2mm}
    \centering
    \subfigure[Normalized throughput vs. $L$]{
        \includegraphics[width=0.47\columnwidth]{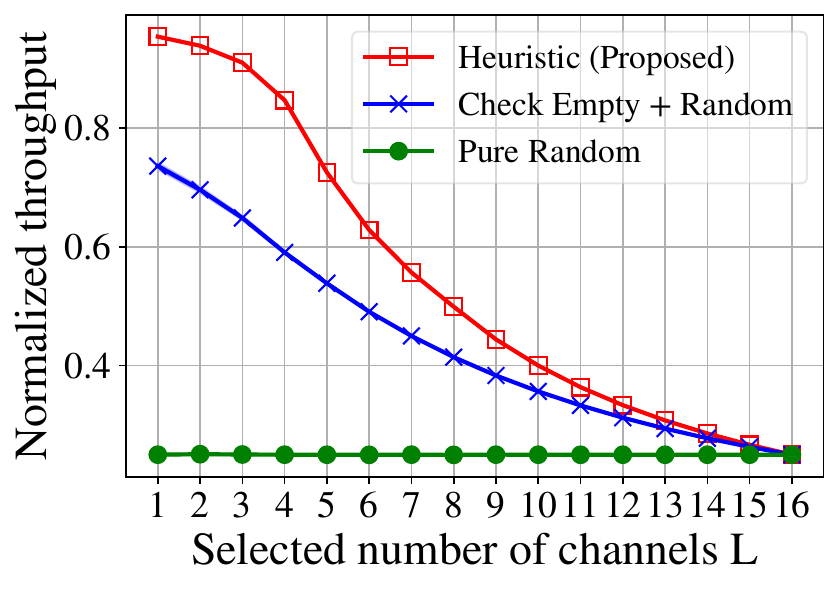}
        \label{fig:Correlated_Varying_L_Throughput}
    }
    \hspace{-10 pt}
    \subfigure[Average collision rate vs. $L$]{
        \includegraphics[width=0.47\columnwidth]{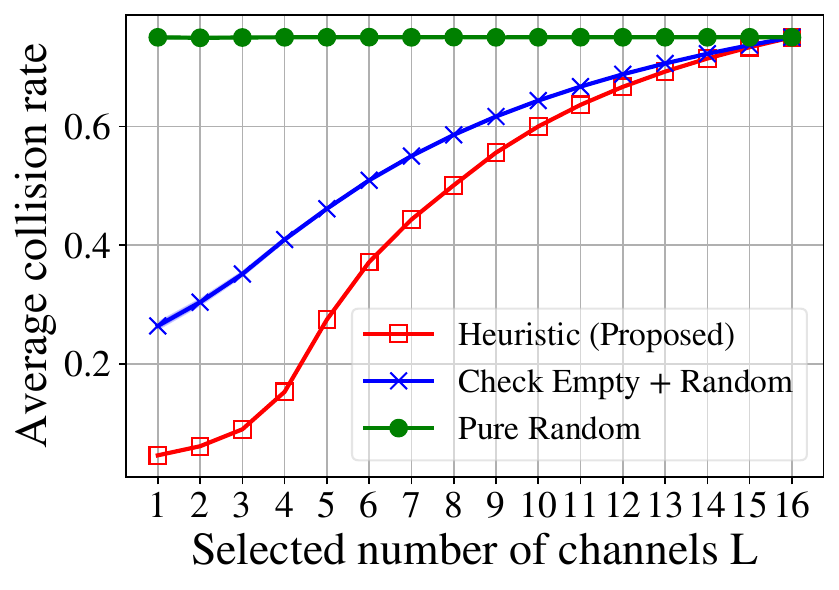}
        \label{fig:Correlated_Varying_L_Loss}
    }
        \vspace{-2mm}
    \caption{Performances of different scheduling policies vs. selected channel number $L$ under correlated channel occupancies.} 
    \label{fig:Correlated_Varying_L}
    \vspace{-2mm}
\end{figure}

\begin{figure}[hbt!]
     \centering
     \subfigure[Normalized throughput vs. $\sigma$]{
         \includegraphics[width=0.47\columnwidth]{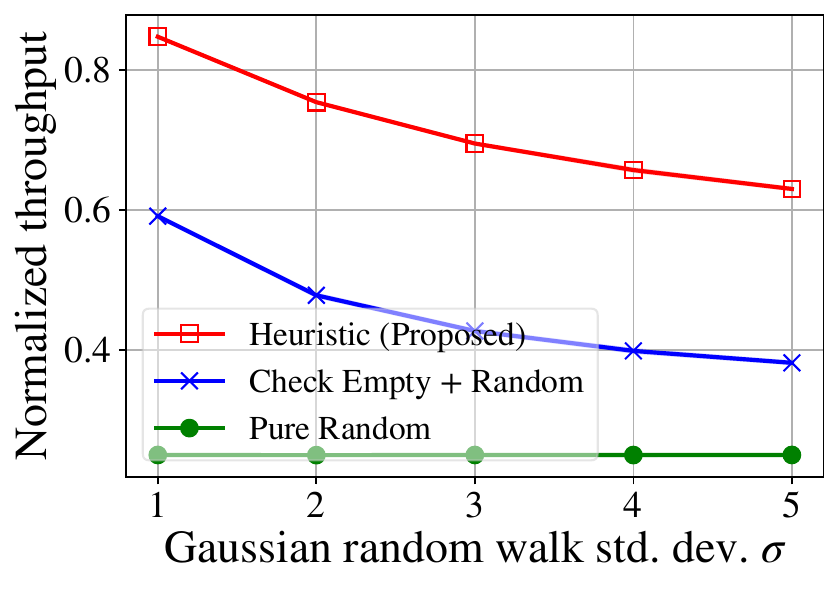}
          \label{fig:Correlated_Varying_sigma_Throughput}
     }
     \hspace{-10 pt}     
     \subfigure[Average collision rate vs. $\sigma$]{
         \includegraphics[width=0.47\columnwidth]{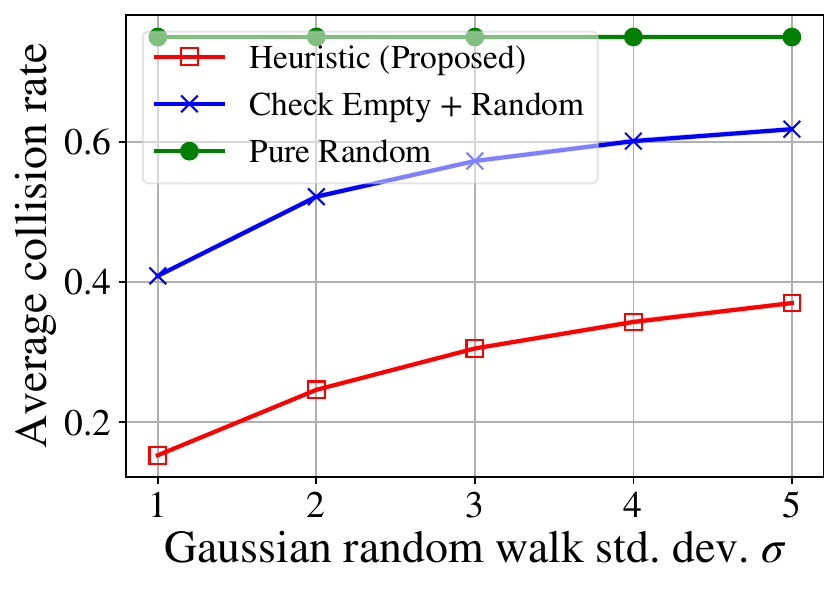}
          \label{fig:Correlated_Varying_sigma_Loss}
     }     
     \hfill
     \subfigure[Average collision rate vs. Normalized throughput for varying $\sigma$]{
          \includegraphics[width=0.85\columnwidth]{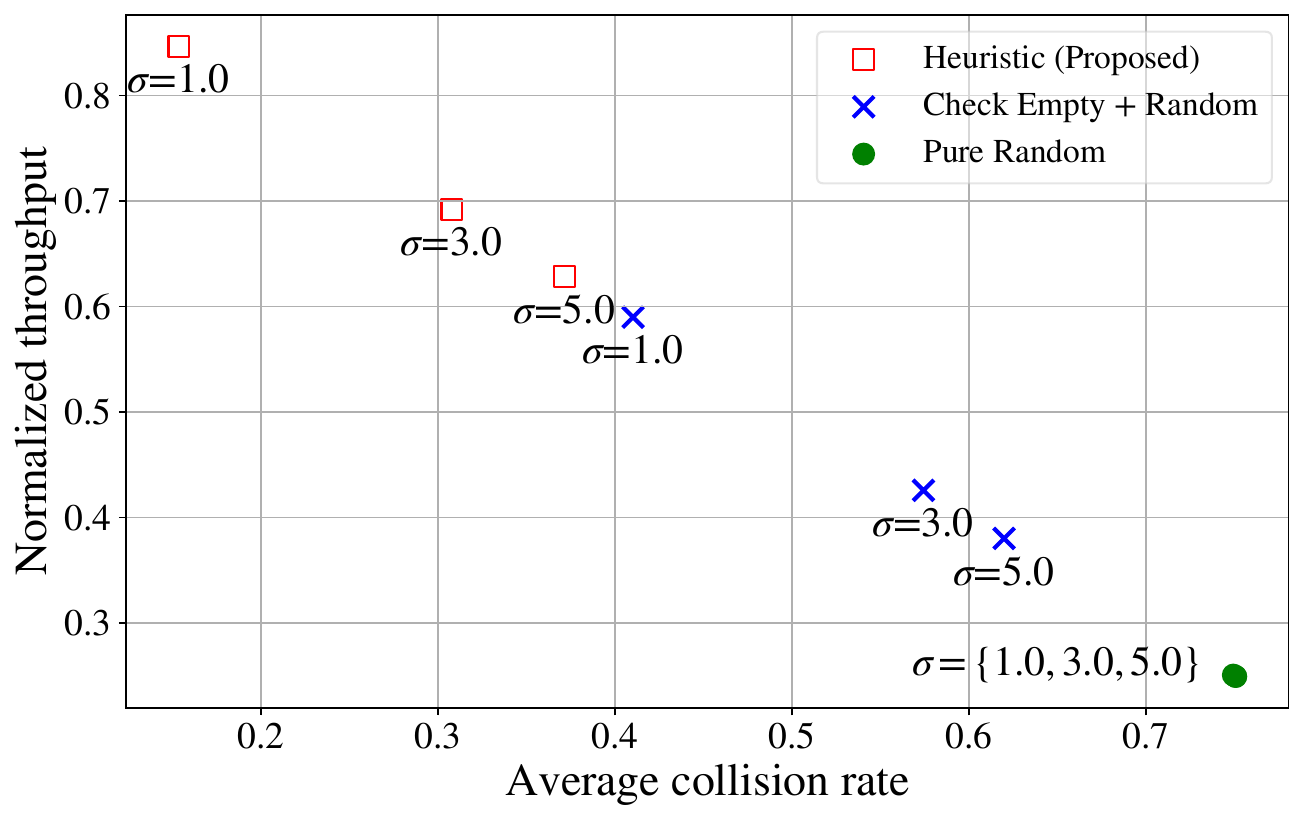}
           \label{fig:Correlated_Varying_sigma_overall}
     }
  
     \vspace{-2mm}
     
     \caption{Performances of different scheduling policies vs. Gaussian random walk standard deviation $\sigma$ (fixed total channel number $N=16$, PU band length $B=12$, and selected channel number $L=4$).}
     \label{fig:Correlated_Varying_sigma}
     
     \vspace{-3mm}
\end{figure}

Next, Figure \ref{fig:Correlated_Varying_sigma} illustrates how changes in Gaussian random walk standard deviation $\sigma$ influence system performance. In Figure \ref{fig:Correlated_Varying_sigma_Throughput}, we observe that increasing the standard deviation $\sigma$ of the Gaussian random walk leads to a gradual degradation in SU throughput for all strategies, since the PU occupancy band moves at a faster rate leading to higher uncertainty. Despite this, our heuristic index policy achieves up to a \textbf{65\% gain in average SU throughput} compared to the Check Empty + Random policy, consistently maintaining values above 0.6 even under increased mobility uncertainty $\sigma$, demonstrating strong robustness of the Heuristic index policy. Figure \ref{fig:Correlated_Varying_sigma_Loss} presents the average collision rate, where the proposed heuristic maintains the lowest collision rate among the all schemes, achieving up to a \textbf{62\% reduction in the collision rate} compared to the Check Empty + Random policy. To better visualize the trade-off between throughput and collision rate, Figure \ref{fig:Correlated_Varying_sigma_overall} plots these two metrics against each other. It clearly highlights that the proposed heuristic index policy achieves a superior balance, sustaining relatively high throughput while keeping the collision rate low even at high $\sigma$ values. 

These results demonstrate that the proposed heuristic index policy motivated by Whittle index also performs well in PU occupancy models that are correlated across channels.

\section{Conclusion} \label{sec:Conclusion}
In this work, we studied the spectrum sharing problem where a secondary user (SU) opportunistically communicates over multiple Markov channels occupied by a primary user (PU). Using this model, we established indexability of the decoupled problem and derived the Whittle index policy, which demonstrates strong performance by achieving both high SU throughput and low collision rate. Furthermore, we extended our approach to accommodate PU occupancy models with inter-channel correlation and to incorporate learning of unknown Markov matrices by proposing a heuristic index function inspired by the Whittle index. 

Important directions of future work involve 1) proving theoretical guarantees for the heuristic index policy in scenarios with correlated channels, 2) considering multiple SUs and competition between them for open channels, and 3) extending data-driven online learning and prediction for more general PU occupancy models.



\bibliographystyle{IEEEtran}
\bibliography{Refs/ALOHA_CSMA, Refs/AoI,Refs/Channel_Model, Refs/Dynamic_Programming, Refs/Spectrum_Sensing, Refs/Virtual_Queue, Refs/Whittle_Index}

\appendices

\section{Proof of Theorem \ref{thm: Decoupled Singe-Channel Problem}:Relaxation and Decoupling}

As a preliminary step, we consider a relaxed version of the original spectrum sharing problem \eqref{eq:P1, Unified}, where the spectrum access constraint is satisfied on average rather than at every time-slot. This leads to the following optimization problem:

\begin{equation} \label{eq:P1-b}
\begin{aligned}
     \max_{\pi} \quad &   \lim \limits_{T \to \infty} \mathbb{E}_{\pi} \left[  
  \frac{1}{L} \frac{1}{T} \sum_{i=1}^{N} \sum_{t=1}^{T} \Big(R_{i}(t) - C_i(t)\Big) \right] \\
    \textrm{s.t.}     
    \quad & \lim \limits_{T \to \infty} \frac{1}{T} \sum_{i=1}^{N}  \sum_{t=1}^{T}  u_i(t) \leq L,\\
    \quad & u_i(t) \in \{0,1\}, \forall i \in [N], \forall t \in [T].
\end{aligned}
\tag{P1-b}
\end{equation}

The key difference between \eqref{eq:P1, Unified} and \eqref{eq:P1-b} is that \eqref{eq:P1-b} has just one spectrum access constraint over the entire time horizon instead of a constraint for every time-slot. This makes it amenable to solving via a Lagrangian relaxation \cite{kriouile2024asymptotically}. Given a Lagrange multiplier $C>0$ for the time-average spectrum access constraint, we can convert \eqref{eq:P1-b} to the following optimization problem:

\begin{equation} \label{eq:P1-c}
\begin{aligned}
     \max_{\pi} \quad &  \lim \limits_{T \to \infty} \mathbb{E}_{\pi} \left[ \frac{1}{T} \sum_{i=1}^{N}\sum_{t=1}^{T} \Big(R_{i}(t) - C_i(t) - C u_i(t) \Big) \right]
     \\
    \textrm{s.t.}     
    \quad & u_i(t) \in \{0,1\}, \forall i \in [N], \forall t \in [T].  
\end{aligned}    
\tag{P1-c}
\end{equation}

Note that in \eqref{eq:P1-c}, the objective function now separates as a sum across each channel without any global constraints that couple decision-making across them. In particular, the Lagrangian multiplier $C$ can be interpreted as a transmission cost for the SU for using a channel. Specifically, whenever the SU selects channel $i$ for transmission ($u_i(t)=1$), it must pay a cost of $C$ for accessing the channel $i$. Based on the definition of SU throughput and collision penalties on channel $i$, we can simplify the objective function for channel $i$ as follows:
\begin{align}
\begin{aligned}
     &R_i(t) - C_i(t) -Cu_i(t) \\
     &= u_i(t)(1-X_i(t)) -\gamma u_i(t) X_i(t)-Cu_i(t) \\
     &= (1+\gamma) R_i(t) - (\gamma+C)u_i(t).    
\end{aligned}
\end{align}

By normalizing with the constant $1+\gamma$ and introducing the effective transmission cost $D=(\gamma+C)/(1+\gamma)$, we obtain the following decoupled problem for a single channel:

\begin{align}
\begin{aligned}
    \max_{\pi} \quad & \lim \limits_{T \to \infty} \mathbb{E}_{\pi} \left[  \frac{1}{T}   \sum_{t=1}^{T} \Big(R_{i}(t) - D u_i(t)\Big) \right] \\   
   \textrm{s.t.}     
   \quad & u_i(t) \in \{0,1\}, \forall t \in [T]. \\  
\end{aligned} 
\tag{P1-d}
\end{align}
\qed

\section{Proof of Lemma \ref{lem:MDP Evolution}: MDP Evolution} \label{app:MDP Evolution}
We prove the result by induction of on $\Aoi_i(t)$. For any  $\Aoi_i(t) \geq 1$, the actual primary user (PU) occupancy at time $t$ is determined by the Markov transition matrix $Q_i$, given by 
\begin{align}
\begin{bmatrix}
\P\{X_i(t) = 0\} \\
\P\{X_i(t) = 1\}
\end{bmatrix}
=
Q_i^{\Aoi_i(t)} \cdot
\begin{bmatrix}
\P\{X_i(t-\Aoi_i(t)) = 0\} \\
\P\{X_i(t-\Aoi_i(t)) = 1\}
\end{bmatrix}.
\end{align}

Expanding the probability expressions component-wise using the Markov transition matrix, we obtain the conditional probability formulation:
\begin{align}
\P\{X_i(t)=0\} = [Q_{i}^{\Aoi_i(t)}]_{00} \cdot \P\{\hatX_i(t) = 0\} \notag \\+ [Q_{i}^{\Aoi_i(t)}]_{01}  \cdot \P\{\hatX_i(t) = 1\}, \\
\P\{X_i(t)=1\} = [Q_{i}^{\Aoi_i(t)}]_{10} \cdot \P\{\hatX_i(t) = 0\} \notag \\+ [Q_{i}^{\Aoi_i(t)}]_{11}  \cdot \P\{\hatX_i(t) = 1\}.
\end{align} 
Since the transition matrix is symmetric, specifically $[Q_{i}^{\Aoi_i(t)}]_{01} = [Q_{i}^{\Aoi_i(t)}]_{10}$ and $[Q_{i}^{\Aoi_i(t)}]_{00} = [Q_{i}^{\Aoi_i(t)}]_{11}$. This completes proof. $\qed$

\section{Proof of Lemma \ref{lem:Bellman equations}: Bellman equations} \label{app:Bellman equations}
The Bellman equation recursion characterizes the optimal decision-making process by balancing immediate rewards and costs with long-term expected profit. It ensures that the optimal strategy at each step is derived by maximizing over all possible decisions while accounting for Markovian state transitions. Let $\lambda$ be the time-average reward under the optimal policy. The differential value function $S(\hatX(t),\Aoi(t))$ quantifies the relative value of being in state $(\hatX(t),\Aoi(t))$ compared to this time-average reward. Observe that if no transmission attempt is made ($u=0$), only the AoI $\Aoi$ increases by 1. In contrast, when transmission is attempted ($u=1$), the immediate reward depends on the PU occupancy $X(t)$: an immediate reward of +1 is received only if the channel is vacant ($X(t)=0$). Regardless of the PU occupancy, the effective transmission cost $D$ is always incurred. Therefore, using Lemma \ref{lem:MDP Evolution}, we obtain the following two recursive equations:
\begin{align}
\lambda + S(1,\Aoi) = \max\limits_{u \in \{0,1\}} \bigg\{& S(1,\Aoi+1), \big[Q^\Aoi\big]_{10}(S(0,1)+1)
\notag \\ + \big[Q^\Aoi\big]_{11}&S(1,1) - D \bigg\}, ~\forall \Aoi \in \mathbb{Z^+}.
\end{align}

\begin{align}
\lambda + S(0,\Delta) = \max\limits_{u \in \{0,1\}} \bigg\{& S(0,\Aoi+1), \big[Q^\Aoi\big]_{00}(S(0,1)+1) \notag \\+ \big[Q^\Aoi\big]_{01}&S(1,1) -D\bigg\}, ~\forall \Aoi \in \mathbb{Z^+}.
\end{align}
\qed

\section{Proof of Lemma \ref{lem:Average Reward λ}: time-average reward} \label{app:Average Reward λ}
Without loss of generality, we set $S(0,1) = S \in \mathbb{R}$, and $S(1,1) = 0$. Assume that the optimal policy follows a threshold structure, i.e., there exists a threshold $H$ such that it is optimal to access the channel ($u = 1$) for all states $\Aoi \geq H$ and to remain idle ($u= 0$), otherwise. Since our optimal state space is restricted, it suffices to examine the Bellman equation \eqref{eq:Bellman Case 1} for $\{(1, \Aoi) | \Aoi \geq 1\}$ to verify the optimality of the proposed threshold policy.

\textbf{Case 1)}: If $\Aoi \geq H$, then the action $u=1$ under the threshold policy and the Bellman equation \eqref{eq:Bellman Case 1} reduce to
\begin{align}
\begin{aligned}
S(1,\Delta) =  &\left[Q^{\Aoi}\right]_{10}\big(S(0,1)+1\big) \\&+ \left[Q^{\Aoi}\right]_{11} S(1,1) -D - \lambda, ~\forall \Delta \geq H.
\end{aligned}
\end{align}
Since $S(0,1) = S\in \mathbb{R} $ and $S(1,1) = 0$, we can simplify this to
\begin{align}  \label{eq:Aoi ≥ H}
S(1,\Aoi) = \left[Q^{\Aoi}\right]_{10}(S+1) -D - \lambda, ~\forall \Aoi \geq H.
\end{align}\\

In addition, we also need to show that $\left[Q^{\Aoi}\right]_{10}$ is a monotone increasing and concave function with respect to $\Aoi$, which implies that $S(1,\Aoi+1) \geq S(1,\Aoi) \text{ for all } \Aoi \geq H$.  Note that every symmetric transition matrix $Q$ with transition probability $q\in(0,1/2)$ can be orthogonally diagonalizable and 

\begin{align} \label{eq:power of Q}
    Q^{\Aoi} &= PD^{\Aoi}P^{-1} = \begin{bmatrix} \frac{1}{\sqrt{2}} & \frac{1}{\sqrt{2}} \\ \frac{1}{\sqrt{2}} & \frac{-1}{\sqrt{2}} \end{bmatrix} \begin{bmatrix} 1 & 0 \\ 0 & (1-2q)^{\Aoi} \end{bmatrix} \begin{bmatrix} \frac{1}{\sqrt{2}} & \frac{1}{\sqrt{2}} \\ \frac{1}{\sqrt{2}} & \frac{-1}{\sqrt{2}} \end{bmatrix}.
\end{align}

As a result, we observe that $\lim_{\Aoi \to \infty}[Q^{\Aoi}]_{01}=\lim_{\Aoi \to \infty}[Q^{\Aoi}]_{00}=1/2 ~\forall q \in (0,1/2)$. This reflects the intuition that as $\Aoi$ becomes large, the effect of the initial state diminishes and the system reaches a steady state, where the transition probability between two states converges to $1/2$. Such behavior aligns with the dynamics of symmetric systems, where transitions between states 0 and 1 tend to become equally likely over time, reflecting a state of maximum entropy. From \eqref{eq:power of Q}, we observe that  $[Q^\Aoi]_{01}=\frac{1-(1-2q)^\Aoi}{2}$. This function is monotonically increasing with respsec to $\Aoi$, as its the first derivative is given by $\frac{d[Q^\Aoi]_{01}}{d\Aoi}=-\frac{1}{2}(1-2q)^{\Aoi}\cdot\ln(1-2q)$ is always positive for all $q \in (0,1/2)$. To establish concavity, we consider the second derivative $\frac{d^2[Q^\Aoi]_{01}}{d^2\Aoi}=-\frac{1}{2}(1-2q)^{\Aoi}\cdot[\ln(1-2q)]^2$, which is always negative for all $q\in(0,1/2)$.

Therefore, our value function $S(1, \Aoi)$ continues to increase until AoI $\Aoi$ reaches the threshold H, reflecting the intuitive idea that once the channel is already occupied, waiting until AoI reaches the threshold is beneficial for maximizing the reward.

\textbf{Case 2)}: If $\Aoi < H$, then the action $u=0$
under the threshold policy and the Bellman equation \eqref{eq:Bellman Case 1}  reduce to
\begin{align} \label{eq:Aoi < H}
S(1,\Delta) =  S(1,\Delta+1) - \lambda ,~\forall \Delta < H.
\end{align}
Using the induction, for a general $k$, we obtain
\begin{align}
\begin{aligned}
S(1,H-k) &= S(1,H)-k\lambda \\
        &= \Big(\left[Q^{H}\right]_{10}(S+1) -D \Big) - (k+1)\lambda. 
\end{aligned}
\end{align}\\
Setting $k = H-1$, we derive
\begin{align}
S(1,1) &= \left[Q^{H}\right]_{10}(S+1) -D - H\lambda.
\end{align}\\
Observe that since we assume S(1,1) = 0, it follows that
\begin{align} \label{eq:S(1,Aoi)}
H\lambda = \left[Q^{H}\right]_{10}(S+1) -D.
\end{align}

Now, referring the Bellman equation \eqref{eq:Bellman Case 2} for the state $(0,\Aoi)=(0, 1)$, and using $\left[ Q \right]_{10}=1-\left[ Q \right]_{00}$, we get
\begin{align}
S= \frac{(1-\left[ Q \right]_{10}) -D -\lambda}{\left[ Q \right]_{10}}. \label{eq:S(0,Aoi)}
\end{align} 
Substituting \eqref{eq:S(0,Aoi)} to \eqref{eq:S(1,Aoi)} and rearranging gives
\begin{align}
\Rightarrow \lambda = \frac{[Q^{H}]_{10}-([Q^{H}]_{10}+[Q]_{10})D}{[Q^{H}]_{10}+H[Q]_{10}},
\end{align}
which completes proof. $\qed$

\section{Proof of Theorem \ref{thm:Optimal threshold H}: optimal threshold} \label{app:Optimal threshold H}


Since the threshold $H$ is defined over discrete time space $\mathbb{Z^+}$, directly finding the optimal value $H^*$ that maximizes the time-average reward $\lambda(H,D)$ is inherently difficult. Even when attempting to approximate the optimal threshold $H^*$ by relaxing $H$ to the continuous domain and differentiating $\lambda(H,D)$ with respect to $H$, the resulting expressions do not admit a closed-form solution due to the involvement of the Lambert $W$-function. Therefore, instead of deriving a closed-form expression for $H$, we determine a necessary condition for $H$ by formulating an inequality for the time-average reward $\lambda$. This result is also used to establish indexability and derive the Whittle Index.  

Recall that the Bellman equation \eqref{eq:Bellman Case 1} for the states $\{(1, \Aoi) | \Aoi \geq 1\}$ is given by 
\begin{align}
S(1,\Aoi) = \max_{u \in\{0,1\}} \{S(1,\Aoi+1), \; [Q^\Aoi]_{10}(S+1)-D \} - \lambda.
\end{align}

Assume that there exists an optimal threshold $H^*$. Then, for the Bellman equation to satisfy the threshold structure, the following conditions must hold to maximize value function $S(\cdot)$.
\begin{equation}
\begin{cases}
\forall \Aoi < H^*,  \; S(1,\Aoi+1) \geq [Q^\Aoi]_{10}(S+1) - D\\
\forall \Aoi \geq H^*, \; S(1,\Aoi+1) \leq [Q^\Aoi]_{10}(S+1) - D.
\end{cases}
\end{equation}

Substituting $\Aoi=H^*-1$ and $\Aoi=H^*$ into these conditions, we obtain

\begin{equation}
\begin{cases}
S(1, H^*) \geq [Q^{H^*-1}]_{10}(S+1) - D\\
S(1, H^*+1) \leq [Q^{H^*}]_{10}(S+1) - D.
\end{cases}
\end{equation}

From \eqref{eq:Aoi ≥ H}, we can simplify this to 

\begin{equation}
\begin{cases}
\big([Q^{H^*}]_{10}-[Q^{H^*-1}]_{10}\big)(S+1) \geq \lambda \\
\big([Q^{H^*+1}]_{10}-[Q^{H^*}]_{10}\big)(S+1) \leq \lambda.
\end{cases}
\end{equation}

Let $g(H):=\big([Q^{H+1}]_{10}-[Q^{H}]_{10}\big)(S+1)$. Thus, we obtain
\begin{equation} \label{eq:Necessary Condition of Optimal Threshold}
g(H^*) \leq \lambda \leq g(H^*-1).
\end{equation}

The final step is to verify the monotonicity of $g(H)$ for the feasibility of \eqref{eq:Necessary Condition of Optimal Threshold}. Given that $[Q^H]_{10}$ is a monotone increasing and concave function on $H$, $[Q^{H+1}]_{10}-[Q^H]_{10}$ becomes a monotone decreasing function. To ensure that $g(H)$ is decreasing, it is sufficient that $S(1,0)=S>-1$. This condition holds necessarily because if $S\leq-1$, the Bellman equation would always prefer idling (u=0), thereby breaking the threshold type structure. Moreover, as $S(0,1)$ represents the value function that is typically non-negative, this assumption is both reasonable and naturally satisfied in practice. Therefore, the optimal threshold $H^*$, obtained from the Bellman equation, guarantees the maximization of the time-average reward. This completes the proof. $\qed$

\section{Proof of Lemma \ref{lem:Indexability} and Theorem \ref{thm:Whittle Index}: indexability} \label{app:Indexability}

The indexability property for the decoupled problem requires that, as the transmission cost $D$ increases from 0 to $\infty$, the set of AoI values for which it is optimal to transmit, i.e. $\mathscr{P}(D) = \{\Aoi \in \mathbb{Z^+}\ \vert \; \Aoi < H^*(D) \}$, must increase monotonically from the empty set $\phi$ (never transmit) to the entire set $\mathbb{Z^+}$ (always transmit). Here, the optimal threshold $H^*(D)$ should increase accordingly as effective transmission cost $D$ increases.

The first step is to derive an inequality related to the transmission cost $D$ from Theorem \ref{thm:Optimal threshold H}. Substitute \eqref{eq:S(0,Aoi)} to Theorem \ref{thm:Optimal threshold H}, and rearranging for $\lambda$, we obtain

\begin{equation}
\begin{cases}
(1+\frac{[Q^{H^*+1}]_{10}-[Q^{H^*}]_{10}}{[Q]_{10}}) \cdot \lambda \geq \frac{[Q^{H^*+1}]_{10}-[Q^{H^*}]_{10}}{[Q]_{10}}(1-D) \\
(1+\frac{[Q^{H^*}]_{10}-[Q^{H^*-1}]_{10}}{[Q]_{10}}) \cdot \lambda  \leq \frac{[Q^{H^*}]_{10}-[Q^{H^*-1}]_{10}}{[Q]_{10}}(1-D).
\end{cases}        
\end{equation}

Since $\lambda$ was obtained in Lemma \ref{lem:Average Reward λ} as a closed-form expression of the effective transmission cost $D$ and the optimal threshold $H^*$, substituting $\lambda(H^*,D)$ into the second inequality yields the upper bound of $C$ given by 

\begin{align}
       & D\leq W(1,H^*) \notag \\ &\triangleq \frac{H^*([Q^{H^*-1}]_{10}-[Q^{H^*}]_{10})+[Q^{H^*}]_{10}}{(H^*-1)([Q^{H^*-1}]_{10}-[Q^{H^*}]_{10})+[Q^{H^*}]_{10}+[Q]_{10}}. \label{eq:38}
\end{align}

The numerator can be rewritten as $N(H^*)=[Q^{H^*}]_{10}-H^*([Q^{H^*}]_{10}-[Q^{H^*-1}]_{10})$. Since $[Q^{H^*}]_{10}$ is concave and increasing in $H^*$, its incremental difference $[Q^{H^*}]_{10}-[Q^{H^*-1}]_{10}$ decreases as $H^*$ increases. Although the coefficient $H^*$ increases linearly, the decreasing nature of the difference term due to concavity ensures that the product $H^*([Q^{H^*}]_{10}-[Q^{H^*-1}]_{10})$ increases slowly or diminishes. As a result, the overall expression for $N(H^*)$ increases with $H^*$, because the increases in $[Q^{H^*}]_{10}$ outweighs the growth in the subtracted term.

Then Equation \eqref{eq:38} can be rewritten as 
\begin{align}
W(1,H^*) 
= \frac{1}{1+\frac{[Q^{H^*}]_{10}-[Q^{H^*-1}]_{10}+[Q]_{10}}{N(H^*)}}.
\end{align}

The term $[Q^{H^*}]_{10}-[Q^{H^*-1}]_{10}+[Q]_{10}$ is decreasing because the incremental difference is decreasing by the concavity of $[Q^{H^*}]_{10}$, and $[Q]_{10}$ is constant. As a result, the overall fraction inside the denominator decreases with $H^*$. Therefore, the entire expression $W(1,H^*)$, which is the reciprocal of a decreasing function, becomes monotone increasing. 


Now, considering the monotonicity of $W(1, H^*)$ on $H^*$, let us prove indexability.

\textbf{Case 1)} If $D=0$, then $\lambda(H,0)=\frac{[Q^H]_{10}}{[Q^H]_{10}+H[Q]_{10}}$. Since the denominator is always greater than the numerator and their difference is proportional to threshold $H$, $\lambda(H)$ decreases monotonically as AoI increases.  Then, $H^{*}(0) = \argmax_{H} \; \lambda(H,0) = 1$, which is the smallest element in the natural numbers, so $\mathscr{P}(D=0)=\{\Aoi\in \mathbb{Z^+} \vert \Aoi < 1\}=\phi$. Since there is no transmission cost for sensing, always transmitting is the optimal threshold policy for maximizing the time-average throughput reward, which is an intuitive result.

\textbf{Case 2)} If $D>0$, we will use the monotonicity of the upper bound $W(H)$.
By the definition, $W(1,1)=0$, thus $W(1,1)<D$. Now, if there exits some $H>1$ such that $W(1,H) \geq D$, then we know that there also exists some optimal threshold $H^*(D)$ such that $W(1,H^*-1)\leq D \leq W(1,H^*)$ using monotonicity of $W(1, H^*)$ with respect to $H^*$. Using this, we can relate the optimal threshold $H^*$ to transmission cost $D$. Let $D$ be such that it lies in the interval $[W(1, H^*-1),\; W(1, H^*)]$, then the optimal policy is of threshold type with the threshold at $H^*$. Observe that if $W(1, H^*)$ is strictly increasing then there can only be one such interval in which $D$ can lie. Therefore, we conclude that an increase in the transmission cost $D$ results in an increase in the optimal threshold $H^*(D)$, leading to an expansion of $\mathscr{P}(D)=\{\Aoi \in \mathbb{Z^+} \vert \Aoi < H^*(D)\}$. When $W(1, \Aoi)<D,~\forall \Aoi \in \mathbb{Z^+}$, we choose $H^*=\infty$ and $\mathscr{P}(D)=\{\Aoi \in \mathbb{Z^+} \vert \Aoi < \infty\}=\mathbb{Z^+}$. This completes the proof of indexability for the decoupled problem.

As the final step, since indexability is guaranteed, the Whittle index can be properly derived. From the definition of the Whittle index, the cost that makes both actions equally desirable for a given state corresponds to what we have already established. We have established that $W(1, \Aoi)$ satisfies all the required condition from the Bellman equation through the proof of indexability. Furthermore, $W(0,\Aoi)$ is defined as $\infty$ due to the fact that transmission always optimal when $\hatX=0$, meaning no finite cost can make both actions equally desirable. Therefore, it follows that $W(\hatX, \Aoi)$ is the Whittle index. $\qed$

\section{Proof of Lemma \ref{lem:Approximated Whittle index}: Heuristic Index} \label{app:Approximated Whittle index}

Let $G_i(t)$ represent the number of packets that the SU can transmit on channel $i$ starting at time-slot $t$ until it faces its first collision.
Then its expectation is given by
\begin{align}
&\E[G_i(t)] =\P\{\text{channel $i$ is free at time-slot $t$} ~|~ \text{AoI } \Aoi_i\} \notag \\
& \times \E[\text{free duration of channel $i$ until next collision}].
\end{align}

From Assumption \ref{assm: Markov transition prob.}, the PU channel occupancy follows a two-state discrete-time Markov chain. When the channel is in the free state, it remains in this state with probability $1-q_i$, and transitions to the occupied state with probability $q_i$ at each time-slot. Hence, the free duration of channel $i$ until next collision, denoted $T_i^{\text{free}}$ follows a Geometric distribution with parameter $q_i$ as
\begin{align}
\P\{T_i^{\text{free}} = k\} = (1-q_i)^{k-1}q_i, ~\forall k \in \mathbb{Z^+}.
\end{align}
Accordingly, the expected free duration is given by
\begin{align}
\E[T_i^{\text{free}}] = q_i\sum_{k=1}^{\infty}k(1-q_i)^{k-1} = \frac{1}{q_i} ~\because 0<q_i<1.
\end{align}

From Lemma \ref{lem:MDP Evolution}, the conditional probability that channel $i$ is free at time-slot $t$, given that it was last observed to be occupied $\Aoi$-slot ago, is
\begin{align}
\P\{X_i(t)=0~|~X_i(t-\Aoi_i)=1\}=[Q_i^{\Aoi_i}]_{10}.
\end{align}

Therefore, the expected number of packets that the SU can transmit before a collision is 
\begin{align}
\E[G_i(t)] = [Q_i^{\Aoi_i}]_{10}\cdot\frac{1}{q_i} = \frac{[Q_i^{\Aoi_i}]_{10}}{[Q_i]_{10}}.
\end{align}
\qed

\section{Developing the Correlated Heuristic Index} \label{app:Correlated index}

We consider a correlated channel environment where a PU continuously occupies a fixed block of contiguous $B$ channels, and the entire block shifts over time, such that the center location of the block follows a Gaussian random walk. We discretize and bound this random walk to ensure that the centre location of the band corresponds to an integer representing a channel from 1 to $N$. This setting captures practical systems where primary users, such as wideband communication systems, typically require contiguous spectrum blocks for stable operation.

Let $w(t)\sim N(0,\sigma^2)$ denote the Gaussian random step at time $t$, where $\sigma$ represents the standard deviation controlling the step size. Then, the center of the PU-occupied block, denoted by $P(t)$, evolves according to
\begin{align}
     P(t+1) = P(t) + w(t), \text{ where } w(t) \sim N(0,\sigma^2).
\end{align}


Since the SU only observes a subset $[Selected] \subset [N]$ consisting of $L$ channels at each time-slot, we do not have perfect knowledge of $P(t)$, and instead estimate it using a simple update rule:
\begin{align}
\begin{aligned}
\hatP(t+1)= \alpha\hatP(t) + (1-\alpha)\E[i~|~i \in [Selected], X_i(t)=1],
\end{aligned}
\end{align}
where $\alpha \in [0,1]$ is a parameter that controls the memory of the estimate, and $[Selected]$ denotes the observed subset of size $L$. 

Since $\hatP(t)$ is computed based on occupancy observations that are stale, so we can use their AoI values $\Aoi_i$ to approximately measure the uncertainty in the PU's band location. We do so using the average of AoIs, i.e. $\bar\Aoi=\E[\Aoi_i~|~i \in [Selected]]$. Consequently, we model the true center position of the PU band as a Gaussian random variable with mean $\hatP(t)$ with variance $\bar\Aoi$, i.e., $P(t)\sim\mathcal{N}\big(\hatP(t),\bar\Aoi\big)$.

The heuristic index $\E[G_i(t)|\text{PU occupancy history}]$ for this correlated scenario requires first estimating the probability that channel $i$ is free at time-slot $t$ given the PU occupancy history. The PU occupancy history is captured through the estimated PU center position $\hatP(t)$, which is inferred from a subset of observed channel states over time. This estimate serves as a sufficient statistic summarizing past PU activities, i.e. $V_i(\text{PU occupancy history}) = V_i(\hatP(t))$.

First, we evaluate the probability that channel $i$ is free at time-slot $t$, which corresponds to the event that channel $i$ lies outside the PU's occupied region $[P(t)-B/2, P(t)+B/2]$. This is given by
\begin{align} \label{eq:Heuristic index-Probability}
\P\{&\text{channel } i \text{ is free at time-slot }t\} \notag\\
&= \P\{ i \notin [P(t)-B/2, P(t)+B/2]\} \notag\\
&=1- \P\{P(t) \in [i-B/2, i+B/2] \} \notag\\
&=1 - \Phi(\frac{i+B/2-\hatP(t)}{\sqrt{\bar\Aoi}}) + \Phi(\frac{i-B/2-\hatP(t)}{\sqrt{\bar\Aoi}}). 
\end{align}

Next, we approximate the expected free duration using the concept of mean first-passage time (MFPT) in a Brownian motion process in an interval \cite{redner2001guide}, which is known to scale quadratically with the distance between the current position and destination. We therefore model the free duration as
\begin{align} \label{eq:Heuristic index-Free duration}
    \E\{\text{channel $i$ free duration until next collision}\} \propto |i-\hatP(t)|^2.
\end{align}

Combining the above two quantities \eqref{eq:Heuristic index-Probability} and \eqref{eq:Heuristic index-Free duration}, the heuristic index for the correlated scenario can be approximated as follows
\begin{align}
&V_i(\text{PU occupancy history})= \notag \\
&\begin{cases}
\infty, \text{ if } \hatX_i = 0\\
\Big[ 1 - \Phi(\frac{i+B/2-\hatP(t)}{\sqrt{\bar\Aoi}}) + \Phi(\frac{i-B/2-\hatP(t)}{\sqrt{\bar\Aoi}})\Big]|i-\hatP(t)|^2, \text{ o.w. .} 
\end{cases}
\end{align}
\qed

Note that we have made multiple approximating assumptions to arrive at this expression. Despite this, the heuristic index significantly outperforms all other benchmarks in simulations. This suggests that even approximate (and possibly data-driven/learning-based) index functions could be a good-proxy for making spectrum access decisions for complex scenarios.

\end{document}